\def \AAMAS {}
\ifdefined \VersionLong
	\newcommand{\LongVersion}[1]{#1}
	\newcommand{\ShortVersion}[1]{}
\else
	\newcommand{\LongVersion}[1]{}
	\newcommand{\ShortVersion}[1]{#1}
\fi

\ifdefined\hideAuthors
    \newcommand{\deleteForAnonymity}[1]{\textcolor{black!75}{\itshape (hidden content for anonymity)}}
\else
    \newcommand{\deleteForAnonymity}[1]{#1}
\fi

\ifdefined\VersionAuthor
    \newcommand{\AuthorVersion}[1]{#1}
\else
    \newcommand{\AuthorVersion}[1]{}
\fi

\ifdefined\AAMAS
\documentclass[sigconf]{aamas} 

\usepackage{balance} % for balancing columns on the final page

\makeatletter
\gdef\@copyrightpermission{
  \begin{minipage}{0.2\columnwidth}
  \end{minipage}\hfill
  \begin{minipage}{0.8\columnwidth}
  \end{minipage}
  \vspace{5pt}
}
\makeatother

\usepackage{amsthm}

\newtheorem{theorem}{Theorem}
\newtheorem{lemma}{Lemma}
\newtheorem{proposition}{Proposition}
\newtheorem{corollary}{Corollary}

\theoremstyle{definition}
\newtheorem{definition}{Definition}
\newtheorem{example}{Example}

\theoremstyle{remark}
\newtheorem{remark}{Remark}

\fi

\ifdefined\KR
\documentclass{article}
\usepackage{kr}

\usepackage{times}
\usepackage{soul}
\usepackage{url}
\usepackage[hidelinks]{hyperref}
\usepackage[utf8]{inputenc}
\usepackage[small]{caption}
\usepackage{epsfig}
\usepackage{calc}
\usepackage{amssymb}
\usepackage{amstext}
\usepackage{amsthm}
\usepackage{multicol, adjustbox}
\usepackage{mathrsfs}
\usepackage{pslatex}
\usepackage{latexsym}
\usepackage{enumitem}
\usepackage{graphicx}
\usepackage{color}
\usepackage{amsthm}
\usepackage{booktabs,tabularx,array}
\theoremstyle{plain}
\newtheorem{lemma}{Lemma}
\newtheorem{proposition}{Proposition}
\newtheorem{theorem}{Theorem}
\newtheorem{corollary}{Corollary}
\newtheorem{assumption}{Assumption}

\theoremstyle{definition}
\newtheorem{definition}{Definition}
\newtheorem{example}{Example}

\theoremstyle{remark}
\newtheorem{remark}{Remark}

\fi

\ifdefined\LLNCS
\documentclass[a4paper,10pt]{llncs}
\makeatletter
\AtBeginDocument{%
  \@ifpackageloaded{hyperref}
  {\def\@doi#1{\href{https://doi.org/#1}
      {\ttfamily https://doi.org/#1}\egroup}}
  {\def\@doi#1{\ttfamily https://doi.org/#1\egroup}}
  \def\doi{\bgroup\catcode`\_=12\relax\@doi}}
\makeatother
\makeatletter
\def\@biblabel#1{[#1]}

\makeatother
\fi
\ifdefined\IEEE
\documentclass[conference]{IEEEtran}
\IEEEoverridecommandlockouts
\usepackage{amsthm}
\usepackage{booktabs,tabularx,array}
\theoremstyle{plain}
\newtheorem{lemma}{Lemma}
\newtheorem{proposition}{Proposition}
\newtheorem{theorem}{Theorem}

\theoremstyle{definition}
\newtheorem{definition}{Definition}
\newtheorem{example}{Example}

\theoremstyle{remark}
\renewenvironment{proof}{\begin{IEEEproof}}{\end{IEEEproof}}

\fi
\ifdefined\VersionWithComments
\fi
\ifdefined\VersionAuthor
\usepackage[backend=biber,backref=true,style=alphabetic,url=false,doi=true,defernumbers=true,sorting=anyt,maxnames=99]%{biblatex}
\renewbibmacro*{doi+eprint+url}{%
	\iftoggle{bbx:doi}
		{\color{black!40}\footnotesize\printfield{doi}}
		{}%
	\newunit\newblock
	\iftoggle{bbx:eprint}
		{\usebibmacro{eprint}}
		{}%
	\newunit\newblock
	\iftoggle{bbx:url}
		{\usebibmacro{url+urldate}}
		{}%
}
\fi
\usepackage[utf8]{inputenc}
\usepackage[english]{babel}
	
\usepackage{amsmath}
\usepackage{mathtools}
\usepackage{breqn}
\usepackage{csquotes}
\usepackage{setspace}
\usepackage{xspace}
\usepackage{varwidth}

\usepackage{algorithm}
\usepackage[noend]{algpseudocode}
\algrenewcommand\algorithmicprocedure{\textbf{function}}

\algnewcommand\algorithmicswitch{\textbf{switch}}
\algnewcommand\algorithmiccase{\textbf{case}}
\algdef{SE}[SWITCH]{Switch}{EndSwitch}[1]{\algorithmicswitch\ #1\ \algorithmicdo}{\algorithmicend\ \algorithmicswitch}%
\algdef{SE}[CASE]{Case}{EndCase}[1]{\algorithmiccase\ #1}{\algorithmicend\ \algorithmiccase}%
\algtext*{EndSwitch}%
\algtext*{EndCase}%

\usepackage[svgnames,table]{xcolor}

\newcommand{\cellHeader}[1]{\cellcolor{gray!20}\textbf{#1}}

\usepackage{subfigure} % subfigure to use IEEE templates

\usepackage[inline]{enumitem}

\definecolor{darkblue}{rgb}{0, 0, 0.7}

\usepackage[fixed]{fontawesome5}
\makeatletter
\def\orcidID#1{\smash{\href{https://orcid.org/#1}{\protect\raisebox{-1.25pt}{\protect\faOrcid{}}}}}
\makeatother

\usepackage[capitalise,english,nameinlink]{cleveref} % load after algorithm2e, amsthm and hyperref
\crefname{line}{\text{line}}{\text{lines}} % to remove the capital

\ifdefined \VersionWithComments
 	\definecolor{colorok}{RGB}{80,80,150}
\else
	\definecolor{colorok}{RGB}{0,0,0}
\fi

\newcommand{\ie}{\textcolor{colorok}{i.e.}\xspace}

\renewcommand{\iff}{\textcolor{colorok}{iff}\xspace}

\usepackage{tikz}
\usetikzlibrary{arrows,arrows.meta,automata,positioning,shapes,calc} % shapes for 'cloud'
\tikzstyle{pta}=[auto, ->, >=stealth']
\tikzstyle{every node}=[initial text=]
\tikzstyle{location}=[rectangle, rounded corners, minimum size=12pt, draw=black, fill=blue!10, inner sep=2pt]
\tikzstyle{invariant}=[draw=black, dotted, inner sep=1pt, node distance=0] % xshift=1em,
\tikzstyle{final}=[double, fill=blue!50]

\usepackage{listings}

\crefname{lstlisting}{Listing}{Listings}
\Crefname{lstlisting}{Listing}{Listings}

\definecolor{codegreen}{rgb}{0,0.6,0}
\definecolor{codegray}{rgb}{0.5,0.5,0.5}
\definecolor{codeblue}{rgb}{0.1,0.1,.8}
\definecolor{codepurple}{rgb}{0.58,0,0.82}
\definecolor{backcolour}{rgb}{0.95,0.95,0.92}

\lstdefinelanguage{etol}{
  sensitive=false,
  morecomment=[l]{//},
  morecomment=[s]{/*}{*/},
  morestring=[b]",
  morestring=[b]',
  morekeywords={Transition, Name_State,Initial_State, Atomic_propositions, Labelling, Number_of_agents, Clocks, Clock_constraints,Invariants, Verification_horizon},
  morekeywords=[2]{},
}

\ifdefined\VersionWithComments
	\usepackage[colorinlistoftodos,textsize=footnotesize]{todonotes}
\else
	\usepackage[disable]{todonotes}
\fi

\ifdefined\LLNCS
    
\else
    
\fi

\ifdefined \VersionWithComments
	\newcommand{\inlinecommentgen}[3]{\mbox{}{\color{#2}{\textbf{#3}\ifx#1\\\else:\ \fi #1}}} % here, ``\\'' stands for ``empty''
	\newcommand{\TODOinline}[1]{\inlinecommentgen{#1}{red}{TODO}}
    \else
	\newcommand{\TODOinline}[1]{}
\fi

\newcommand{\N}{\ensuremath{\mathbb{N}_{\ge 0}}}
\newcommand{\Z}{\ensuremath{\mathbb{Z}}}
\newcommand{\Rplus}{\ensuremath{\mathbb{R}_{\ge 0}}}

\newcommand{\setX}{\ensuremath{\mathbb{X}}}
\newcommand{\setY}{\ensuremath{\mathbb{Y}}}

\newcommand{\set}[1]{\ensuremath{\left\{#1\right\}}}
\newcommand{\Set}[2]{\ensuremath{\left\{#1 \ \Bigm| \ #2  \right\}}}

\newcommand{\PSPACE}{\ensuremath{\textsc{PSPACE}}}
\newcommand{\NPSPACE}{\ensuremath{\textsc{NPSPACE}}}
\newcommand{\coPSPACE}{\ensuremath{\textsc{coPSPACE}}}

\newcommand{\Alphabet}{\ensuremath{\Sigma}}

\newcommand{\Paths}{\ensuremath{\mathit{Paths}}}
\newcommand{\run}{\ensuremath{\rho}}

\newcommand{\last}{\ensuremath{\text{last}}}

\newcommand{\Clocks}{\ensuremath{\mathbb{C}}}
\newcommand{\ClockCard}{H} % cardinality of clocks
\newcommand{\clock}{\ensuremath{x}}
\newcommand{\clocky}{\ensuremath{y}}
\newcommand{\clocki}[1]{\ensuremath{\clock_{#1}}}
\newcommand{\clockval}{\ensuremath{\nu}}

\newcommand{\ClocksZero}{\ensuremath{\vec{0}}} % \overrightarrow{0
\newcommand{\resets}{\ensuremath{R}}
\newcommand{\reset}[2]{\ensuremath{[#1]_{#2}}}

\newcommand{\guard}{\ensuremath{g}}
\newcommand{\compOp}{\bowtie}

\newcommand{\Constraints}{\ensuremath{\mathbb{G}}}

\newcommand{\Atomic}{\ensuremath{AP}}
\newcommand{\atomic}{\ensuremath{p}}

\newcommand{\TA}{\ensuremath{\mathcal{A}}}
\newcommand{\loc}{\ensuremath{\ell}} % location % NOTE: nicer display than ``l''
\newcommand{\loci}[1]{\ensuremath{\loc_{#1}}}
\newcommand{\locinit}{\loci{0}}
\newcommand{\Loc}{\ensuremath{\mathbb{L}}} % set of locations
\newcommand{\Transitions}{\ensuremath{\mathbb{T}}}
\newcommand{\transitionArrow}[1]{\ensuremath{\xrightarrow{#1}}}
\newcommand{\Invariant}{\ensuremath{\mathit{Inv}}}
\newcommand{\Label}{\ensuremath{\mathit{Lab}}}
\newcommand{\Final}{\ensuremath{\mathbb{F}}}

\newcommand{\action}{\ensuremath{a}}
\newcommand{\actioni}[1]{\ensuremath{\action_{#1}}}

\newcommand{\dur}{\ensuremath{\text{dur}}}
\newcommand{\equivt}{\ensuremath{\sim^t}}

\newcommand{\ttz}{\ensuremath{\mathcal{Z}}}
\newcommand{\tts}{\ensuremath{\mathcal{T}}}

\newcommand{\states}{\ensuremath{S}}
\newcommand{\statein}{\ensuremath{s}}
\newcommand{\statei}[1]{\ensuremath{s_{#1}}}
\newcommand{\initialstate}{\ensuremath{\statein_0}}
\newcommand{\finalstates}{\ensuremath{\states_F}}

\newcommand{\transitions}{\ensuremath{E}}

\newcommand{\zonegraph}{\ensuremath{\mathcal{G}}}

\newcommand{\zone}{\ensuremath{Z}}
\newcommand{\symZone}{\ensuremath{\mathcal{\zone}}}

\newcommand{\timeelpase}[1]{\ensuremath{#1^{\uparrow}}}
\newcommand{\timeelpaseof}[2]{\ensuremath{#1^{\uparrow #2}}}

\newcommand{\ETOL}{\textsf{ETOL}\xspace}
\newcommand{\prop}{\ensuremath{\varphi}}
\newcommand{\propOther}{\ensuremath{\psi}}
\newcommand{\ClocksFormula}{\ensuremath{\mathbb{J}}}
\newcommand{\clockformula}{\ensuremath{j}}

\newcommand{\clockvalformula}{\ensuremath{\mu}}

\newcommand{\secretProp}{\ensuremath{\psi_\text{priv}}}

\newcommand{\ETforall}{\ensuremath{
  \mathbin{%
    \raisebox{0.2ex}{\scalebox{0.95}{$\bigcirc$}}%
    \kern-0.65em
    \raisebox{0.1ex}{\scalebox{0.7}{$\forall$}}%
  }%
  \,
}}

\newcommand{\ETexists}{\ensuremath{
    \mathbin{
        \raisebox{0.2ex}{\scalebox{0.95}{$\bigcirc$}}%
        \kern-0.65em
        \raisebox{0.1ex}{\scalebox{0.7}{$\exists$}}%
    }
    \,
}}

\newcommand{\sub}{\ensuremath{\mathbf{Sub}}}
\newcommand{\sat}{\ensuremath{\mathbf{Sat}}}

\def\next{\mathsf{X}\,}
\def\until{\, \mathsf{U} \,}

\def\release{\, \mathsf{R}\,}

\def\globaly{\mathsf{G}\,}

\def\eventualy{\mathsf{F}\,}

\newcommand{\clockMC}{\ensuremath{\clock_0}}
\newcommand{\Pre}{\ensuremath{\mathit{Pre}}}

\newcommand{\Post}{\ensuremath{\mathit{Post}}}

\newcommand{\PreExist}{\ensuremath{\Pre_{\ETexists}}}
\newcommand{\PreForall}{\ensuremath{\Pre_{\ETforall}}}

\newcommand{\procedureOUexists}{\ensuremath{\mathbf{OpU}_{\ETexists}}}
\newcommand{\procedureOUforall}{\ensuremath{\mathbf{OpU}_{\ETforall}}}

\newcommand{\matchOU}{\ensuremath{\mathbf{MatchU}}}
\newcommand{\procedureORexists}{\ensuremath{\mathbf{OpR}_{\ETexists}}}
\newcommand{\procedureORforall}{\ensuremath{\mathbf{OpR}_{\ETforall}}}

\newcommand{\matchOR}{\ensuremath{\mathbf{MatchR}}}

\usepackage[upgrade=true]{acro}
\newcommand{\firststyle}[1]{\emph{#1}}
\acsetup{
	case-insensitive=true,
	single = {true},
	format/first-long = {\firststyle}
}

\DeclareAcronym{cps}{
    short = {CPS},
    long = {Cyber-Physical System}
}
\DeclareAcronym{rts}{
    short = {RTS},
    long = {Real-Time System}
}
\DeclareAcronym{fa}{
    short = {FA},
    long  = {finite automaton},
    long-plural-form = {Finite Automata},
}
\DeclareAcronym{des}{
    short = {DES},
    long  = {discrete event system},
    long-plural-form = {Discrete Event Systems},
}
\DeclareAcronym{ta}{
    short = {TA},
    long  = {timed automaton},
    long-plural-form = {Timed Automata},
    cite = {AlurD1994},
}
\DeclareAcronym{era}{
    short = {ERA},
    long  = {event-recording automaton},
    indefinite = {an},
    long-plural-form = {Event-Recording Automata},
    cite = {AFH99},
}
\DeclareAcronym{rta}{
    short = {RTA},
    long  = {Real-Time Automaton},
    long-plural-form = {Real-Time Automata},
}
\DeclareAcronym{tts}{
    short = {TTS},
    long  = {Timed Transition System},
}
\DeclareAcronym{zonegraph}{
    short = {ZG},
    long  = {Zone Graph},
}
\DeclareAcronym{dbm}{
    short = {DBM},
    long  = {Difference Bound Matrix},
    long-plural-form = {Difference Bound Matrices},
    cite = {BYJ04}
}
\DeclareAcronym{mdp}{
    short = {MDP},
    long  = {Markov Decision Process},
    long-plural-form = {Markov decision processes},
}

\DeclareAcronym{etopaque}{
    short = {ET-opaque},
    long  = {Execution-Time Opaque},
    indefinite = {an},
    post={\acuse{etopacity}}
}
\DeclareAcronym{etopacity}{
    short = {ET-opacity},
    long  = {Execution-Time Opacity},
    indefinite = {an},
    post={\acuse{etopaque}}
}
\DeclareAcronym{etol}{
    short = {ETOL},
    long  = {Execution-Time Opacity Logic},
    indefinite = {an},
}
\DeclareAcronym{tctl}{
    short = {TCTL},
    long  = {Timed Computation Tree Logic},
    cite = {Alur1999},
}

\ifdefined \VersionWithComments
 	\definecolor{colorabbrv}{RGB}{0,100,250}
\else
	\definecolor{colorabbrv}{RGB}{0,0,0}
\fi
\newcommand{\abbrv}[1]{\textcolor{colorabbrv}{#1}}

\ifdefined\KR
\ifdefined\hideAuthors
    \author{Anonymous author(s)}
\else
    \author{%
        Jean Leneutre\inst{3}%
        \raisebox{1ex}{\scalebox{0.8}{\orcidID{0000-0003-1943-1583}}}
        \and
    	Dylan Marinho\inst{2}%
    	\raisebox{1ex}{\scalebox{0.8}{\orcidID{0000-0002-2548-6196}}}
    	\and
        Vadim Malvone\inst{3}%
         \raisebox{1ex}{\scalebox{0.8}{\orcidID{0000-0001-6138-4229}}}
        \and \\
        James Ortiz\inst{1}%
        \raisebox{1ex}{\scalebox{0.8}{\orcidID{0000-0001-5407-963X}}}
    }
    \institute{%
        Université Paris Est Créteil, LACL, Créteil, France
        \and
        Sorbonne Université, CNRS UMR 7606, LIP6, Paris, France
        \and
        LTCI, Télécom Paris, Institut Polytechnique de Paris, Palaiseau, France
    }
\fi
\fi

\ifdefined\LLNCS
\ifdefined\hideAuthors
    \author{Anonymous author(s)}
    \institute{}
\else
    \author{%
Jean Leneutre\inst{3}%
\raisebox{1ex}{\scalebox{0.8}{\orcidID{0000-0003-1943-1583}}}
\and
Dylan Marinho\inst{2}%
\raisebox{1ex}{\scalebox{0.8}{\orcidID{0000-0002-2548-6196}}}
\and
Vadim Malvone\inst{3}%
\raisebox{1ex}{\scalebox{0.8}{\orcidID{0000-0001-6138-4229}}}
\and \\
James Ortiz\inst{1}%
\raisebox{1ex}{\scalebox{0.8}{\orcidID{0000-0001-5407-963X}}}
    }
    \institute{%
        Université Paris Est Créteil, LACL, Créteil, France
        \and
        Sorbonne Université, CNRS UMR 7606, LIP6, Paris, France
        \and
        LTCI, Télécom Paris, Institut Polytechnique de Paris, Palaiseau, France
    }
\fi
\fi
\ifdefined\IEEE
\ifdefined\hideAuthors
    \author{Anonymous author(s)}
\else
    \makeatletter % changes the catcode of @ to 11
    \newcommand{\linebreakand}{%
    \end{@IEEEauthorhalign}
    \hfill\mbox{}\par
    \mbox{}\hfill\begin{@IEEEauthorhalign}
    }
    \makeatother % changes the catcode of @ back to 12

    \author{%
        \IEEEauthorblockN{
            James Ortiz\orcidID{0000-0001-5407-963X}
        }
    	\IEEEauthorblockA{%
    	    \textit{Université Paris-Est Créteil, LACL,}\\\textit{Paris, France}
    	}
    	\and
    	\IEEEauthorblockN{
            Dylan Marinho\orcidID{0000-0002-2548-6196}
        }
    	\IEEEauthorblockA{%
    	    \textit{Sorbonne Université, CNRS UMR 7606, LIP6,}\\\textit{Paris, France}
    	}
        \linebreakand

    	\IEEEauthorblockN{
            Vadim Malvone\orcidID{0000-0001-6138-4229}
        }
    	\IEEEauthorblockA{%
    	    \textit{LTCI, Télécom Paris, Institut Polytechnique de Paris,}\\\textit{Palaiseau, France}
    	}
        \and
    	\IEEEauthorblockN{
            Jean Leneutre\orcidID{0000-0003-1943-1583}
        }
    	\IEEEauthorblockA{%
    	    \textit{LTCI, Télécom Paris, Institut Polytechnique de Paris,}\\\textit{Palaiseau, France}
    	}%
        \linebreakand
    }
\fi
\fi
\ifdefined\AAMAS
\ifdefined\hideAuthors
    \author{Anonymous author(s)}
\else
    \author{%
        Jean Leneutre%
        \raisebox{1ex}{\scalebox{0.8}{\orcidID{0000-0003-1943-1583}}}}
        \affiliation{
          \institution{LTCI, Télécom Paris, Institut Polytechnique de Paris}
           \city{Palaiseau}
           \country{France}}
            \email{jean.leneutre@telecom-paris.fr}

    	\author{Dylan Marinho%
    	\raisebox{1ex}{\scalebox{0.8}{\orcidID{0000-0002-2548-6196}}}}
          \affiliation{
          \institution{Sorbonne Université, CNRS UMR 7606, LIP6}
           \city{Paris}
           \country{France}}
            \email{dylan.marinho@lip6.fr}

        \author{Vadim Malvone%
         \raisebox{1ex}{\scalebox{0.8}{\orcidID{0000-0001-6138-4229}}}}
           \affiliation{
          \institution{LTCI, Télécom Paris, Institut Polytechnique de Paris}
           \city{Palaiseau}
           \country{France}}
            \email{vadim.malvone@telecom-paris.fr}

        \author{James Ortiz%
        \raisebox{1ex}{\scalebox{0.8}{\orcidID{0000-0001-5407-963X}}}}
         \affiliation{
          \institution{Université Paris Est Créteil}
           \city{Créteil}
           \country{France}}
            \email{james.ortiz-vega@u-pec.fr}
    
\fi
\fi

\title{Execution-Time Opacity Logic: A Logic for Ensuring ET-Opacity in Timed Systems%    
}
\begin{abstract}
    Ensuring confidentiality in Cyber-Physical Systems is critical, especially when attackers exploit execution times to infer sensitive information. Traditional opacity models are inadequate for timed systems, as verifying opacity in Timed Automata is undecidable. To address this challenge, we propose \emph{Execution-Time Opacity Logic} (\ETOL), a new formalism that specifies opacity by requiring that for every execution satisfying a secret formula, there exists another execution of the same duration that does not satisfy it. \ETOL guarantees that timing observations cannot reveal confidential agent activities. We present a decidable and efficient verification framework based on zone-based model checking, supported by a dedicated algorithm that systematically identifies duration-equivalent executions. Our approach is validated through an \textsf{ATM} case study, showing that \ETOL enables efficient verification of execution-time confidentiality under timing attacks. We also developed a prototype tool for the \ETOL logic that supports symbolic model checking over timed systems. It allows users to verify \ETOL formulas based on clock-constrained execution paths.
	
\end{abstract}

\newcommand{\BibTeX}{\rm B\kern-.05em{\sc i\kern-.025em b}\kern-.08em\TeX}

\begin{document}

\maketitle

\ifdefined\IEEE
\fi

%%%%%%%%%%%%%%%%%%%%%%%%%%%%%%%%%%%%%%%%%%%%%%%%%%%%%%%%%%%%%%%%%%%%%%
%%%%%%%%%%%%%%%%%%%%%%%%%%%%%%%%%%%%%%%%%%%%%%%%%%%%%%%%%%%%%%%%%%%%%%
\section{Introduction}
%%%%%%%%%%%%%%%%%%%%%%%%%%%%%%%%%%%%%%%%%%%%%%%%%%%%%%%%%%%%%%%%%%%%%%
%%%%%%%%%%%%%%%%%%%%%%%%%%%%%%%%%%%%%%%%%%%%%%%%%%%%%%%%%%%%%%%%%%%%%%

In the modern digital age, the increasing complexity of distributed infrastructures and \acp{cps} has introduced new verification challenges in cybersecurity.
Critical systems such as those in healthcare, transportation, and industrial control must operate under strict time constraints while preserving the confidentiality of sensitive information.
Cybersecurity is a key domain where this need becomes particularly acute.
As attackers grow more sophisticated, it is essential not only to defend against direct attacks but also to ensure that confidential data remains hidden even through system observations.
This security property is known as \emph{opacity}~\cite{Dubreil2010,Ma2017,BOJ15}, where it refers to the inability of an external observer to deduce secret information from event sequences. In critical \acp{rts}, however, the timing of events introduces a new threat dimension: attackers may exploit execution durations to infer secrets that would otherwise remain hidden.
This has led to the development of \emph{timed opacity}, which studies how timing observations can compromise confidentiality. A challenge in \acp{rts} is that attacks are not limited to observing system states and their timestamps, but may also leverage execution times themselves to extract secrets~\cite{KOC96}.
For example, side-channel attacks use the execution time of specific operations to infer secret information, such as cryptographic keys, by measuring the time differences between two operations.
Even when sensitive information is protected by traditional cryptographic mechanisms, attackers may infer it by analysing variations in execution time. For example, Brumley and Boneh demonstrated in \cite{BB05} that an attacker could recover a 1024-bit RSA private key from an OpenSSL server by measuring network response times, using approximately one million queries over about two hours. 

\noindent $\mathsf{Timed \ Opacity.}$ Timed opacity for \acp{ta} was first introduced by Cassez~\cite{Cas09}, defining it as the inability to deduce secret actions based on an observation of a timed trace.
However, he showed that verifying timed opacity for a general \ac{ta} is undecidable, even in (very) restricted subclasses, due to the complexity of the timed language inclusion problem.
Subsequent works proposed partial solutions. %, they are summarized in~\cref{tab-opacity}.
%\vspace{-0.3cm}
\begin{description}
\item[Restricting the model] Some approaches limit the expressiveness of the system model to regain decidability. For instance, \cite{ADL24} identified decidability in certain contexts, such as one-clock \ac{ta} without $\epsilon$-transitions. Others focus on subclasses of \ac{ta}, like \acp{rta}\footnote{which can be
seen as a subclass of \acp{ta} with a single clock, reset at each transition}, and restrict the notion of opacity to initial state opacity~\cite{WangZ17}.
\item[Restricting the notion of opacity] Some works restrict the notion of opacity itself, for example by considering only \emph{initial-state opacity}~\cite{WangZ17}, which asks whether an observer can infer that the execution started from a secret initial state (or, more generally, from a secret set of initial states). This restriction often simplifies the verification problem, but it may fail to capture richer confidentiality properties arising from secrets revealed only later during the execution.
\item[Restricting the attacker's capabilities] Another approach is to limit the attacker’s observational power, for instance by bounding the time window within which observations can be made or compared~\cite{AmmarTYM21}. Such restrictions often make the verification problem more tractable, although they may also weaken the attacker model.
\item[Restricting the observation model] Instead of using traditional observability functions, which often lead to undecidability, some approaches adopt alternative observation models. For example, \cite{ADJ22} introduces a model based on \emph{duration observations} defining \ac{etopacity}, where the observer measures only the total execution time to a final state, without access to trace events. This model brings decidability while addressing realistic side-channel threats, as execution times can be a sufficient source of information leakage.
\end{description}
%\vspace{-0.1cm}
However, all of these approaches focus on location-based opacity, where the secret is defined as the visit of a particular control location.
We argue that this definition is too restrictive for many real-world scenarios, where secrets may depend on complex combinations of control states and timing conditions.
% \dm{Proposition de tableau, à améliorer !, ou voir si on garde}
% \begin{table}
% \begin{tabular}{l|p{.3\linewidth}|p{.2\linewidth}|p{.2\linewidth}|p{.2\linewidth}}
% \cellHeader{Work} & \cellHeader{Formalism} & \cellHeader{Secret} & \cellHeader{Observation} & \cellHeader{Decidability} \\ \hline
% \cite{Cas09} & \ac{TA} and subclasses & Location visit & Timed trace & Undecidable \\ \hline
% \cite{AmmarTYM21} & time-bounded \ac{TA} &  & & \\ \hline
% \cite{ADL24} & one-clock \ac{TA} without $\epsilon$-transitions (\eg) & Location visit & Timed trace & Decidable \\ \hline
% \cite{WangZ17} & & initial-state & & Decidable \\ \hline
% \cite{ADJ22} & \ac{TA} & Location visit & Execution time & Decidable \\ \hline
% \ETOL & \ac{TA} & Satisfaction of a property & Execution time & Decidable \dm{ref theorem}
% \end{tabular}
% \caption{Comparison of timed opacity definitions}
% \label{tab-opacity}
% \end{table}
%\vspace{-0.3cm}

\noindent  $\mathsf{Our \ contribution.}$ We introduce \ETOL, a branching-time  temporal logic  for specifying  and verifying \ac{etopacity} properties with timing constraints. We define the syntax and semantics of \ETOL{}, and enrich its branching-time  framework with two dedicated modalities: $\ETexists{}$ and $\ETforall{}$, which quantify over runs that are indistinguishable to an attacker based on their total execution time. 
Principally, \ETOL{} allows us to express that a secret property (\ie{} its satisfaction) is not revealed to an attacker by the total execution time. From a theoretical perspective, we propose a symbolic procedure for bounded model checking of \ETOL{} over \acp{TA}, where the
execution-time observation is restricted to a finite verification horizon, and establish a \PSPACE{} upper bound for this problem. %we propose a \PSPACE{} symbolic procedure for model checking \ETOL{} over \iac{ta}.

\noindent Compared with prevoiys approaches, our framework allows the secret to be specified as the satisfaction of a temporal property. Moreover, the line of research of ET-opacity \cite{ALMS22,ALLMS23} relies on parameter synthesis procedures to solve the problem of timed opacity (\ie{} the authors propose an undecidable resolution to a decidable problem), while we propose a model checking procedure for a logic that can express timed opacity properties.

\noindent $\mathbf{Related \ work.}$ In addition to the works on timed opacity mentioned above, there are several related lines of research.
Related real-time security lines focus on non-interference: \cite{NNV17} proposes a type system handling nondeterminism and timing, \cite{VNN18} defines information-flow-based bisimulation and computes local security constraints, and \cite{GSB18} reduces verification to unreachability checked with \abbrv{UPPAAL}~\cite{behrmann2006tutorial}.
Control synthesis for strong non-deterministic non-interference in \ac{TA} is studied in~\cite{BCLR15}.
Recently, \cite{WA24} extends \abbrv{PTCTL} (a parametric extension of \ac{tctl}) with hyper-properties, allowing to express properties over multiple runs. Their models and results are not really comparable to ours, as they focus on parametric hyper-properties, while we focus on a specific class of hyper-properties related to execution time without any parameter. In particular, there is no inclusion between our logics \ETOL{} and their logic, and neither logic subsumes the other.
More broadly, opacity has been extensively studied in the setting of Discrete-Event Systems (\abbrv{DES}), using both automata and Petri nets. Classical language-based notions include \emph{initial-state opacity}, where the observer should not be able to infer that the system started from a secret initial state, \emph{current-state opacity}, where the observer should not be able to determine that the current state is secret, and \emph{$k$-step opacity}, where secrecy must remain hidden over a bounded observation horizon~\cite{MAL04,Bryans2005,Bryans2008}. These notions form the standard semantic basis of opacity.
They have given rise to a variety of verification and enforcement techniques. Automata-based opacity verifiers decide whether observations remain compatible with both secret and non-secret behaviors~\cite{Saboori2013}, while Petri net-based analyses address concurrent systems where synchronization is essential~\cite{Bryans2005}. These works provide the general conceptual background of opacity, whereas our contribution focuses on ET-opacity in timed systems and on its logical and symbolic verification.
%\vspace{-0.7cm}

\noindent $\mathbf{Outline.}$ The paper is structured as follows.
\cref{sec:background} recalls the necessary theoretical background.
\cref{sec:etol} introduces the syntax and semantics of \ETOL, with the extension of the definition of \ac{etopacity}, and illustrates its expressiveness with examples. In \cref{sec:mc}, we present our bounded model checking algorithm for \ETOL over \ac{ta}, prove its correctness, and establish its \PSPACE \ complexity.
The implementation of our model checking procedure and its evaluation on a case study of an ATM system are described in \cref{sec:cs-implementation}.
Finally, we conclude in \cref{sec:conclusion} and discuss future research directions. %Due to space limitations, proofs have been moved to the appendix.

%%%%%%%%%%%%%%%%%%%%%%%%%%%%%%%%%%%%%%%%%%%%%%%%%%%%%%%%%%%%%%%%%%%%%%
%%%%%%%%%%%%%%%%%%%%%%%%%%%%%%%%%%%%%%%%%%%%%%%%%%%%%%%%%%%%%%%%%%%%%%
\section{Preliminaries} %\section{General Concepts}
\label{sec:background}
%%%%%%%%%%%%%%%%%%%%%%%%%%%%%%%%%%%%%%%%%%%%%%%%%%%%%%%%%%%%%%%%%%%%%%
%%%%%%%%%%%%%%%%%%%%%%%%%%%%%%%%%%%%%%%%%%%%%%%%%%%%%%%%%%%%%%%%%%%%%%
This section introduces the main notations and formalisms. Let $\Rplus$, $\Z$, and $\N$ denote the sets of non-negative reals, integers, and non-negative integers, respectively. For a set $\setX$, $|\setX|$ denotes its cardinality. For sets $\setX$ and $\setY$, we write $\setX\cap\setY$, $\setX\cup\setY$, $\setX\setminus\setY$, and $\setX\times\setY$ for intersection, union, set difference, and Cartesian product. Inclusion and strict inclusion are denoted by $\setX\subseteq\setY$ and $\setX\subset\setY$, and the empty set by $\emptyset$.
%%%%%%%%%%%%%%%%%%%%%%%%%%%%%%%%%%%%%%%%%%%%%%%%%%%%%%%%%%%%
% \subsection{Clocks, Guards and Invariants}\label{ss:clocks}
%%%%%%%%%%%%%%%%%%%%%%%%%%%%%%%%%%%%%%%%%%%%%%%%%%%%%%%%%%%%
\noindent Let $\Clocks=\{\clocki{1},\dots,\clocki{\ClockCard}\}$ be a finite set of \emph{clocks}, i.e., real-valued variables evolving synchronously at the same rate. A clock valuation is a function $\clockval:\Clocks\to\Rplus$, and $\ClocksZero$ denotes the valuation assigning $0$ to every clock. For $d\in\Rplus$, $\clockval+d$ is defined by $(\clockval+d)(\clock)=\clockval(\clock)+d$ for all $\clock\in\Clocks$. Given $\resets\subseteq\Clocks$, the reset of $\clockval$ with respect to $\resets$, denoted $\reset{\clockval}{\resets}$, is defined by $\reset{\clockval}{\resets}(\clock)=0$ if $\clock\in\resets$, and $\reset{\clockval}{\resets}(\clock)=\clockval(\clock)$ otherwise. A \emph{clock guard} $\guard$ is a conjunction of constraints of the form $\clock\compOp d$, where $\clock\in\Clocks$, $d\in\Z$, and ${\compOp}\in\{<,\leq,=,\geq,>\}$. The set of all clock guards is denoted $\Constraints(\Clocks)$. For a guard $\guard$, we write $\clockval\models\guard$ when replacing each clock $\clock$ by $\clockval(\clock)$ makes $\guard$ true. We write $\llbracket \guard\rrbracket=\{\clockval\mid \clockval\models\guard\}$ for the set of valuations satisfying $\guard$.

%%%%%%%%%%%%%%%%%%%%%%%%%%%%%%%%%%%%%%%%%%%%%%%%%%%%%%%%%%%%
\subsection{Timed Automata and their Semantics}\label{ss:ta}
%%%%%%%%%%%%%%%%%%%%%%%%%%%%%%%%%%%%%%%%%%%%%%%%%%%%%%%%%%%%%
We now recall \acp*{ta}, which extend \acp{fa} with clocks that evolve synchronously and measure time delays.

\begin{definition}[\Aclp{ta} \cite{AlurD1994}]
\label{def:wta}
A \ac{ta} is a tuple $\TA = (\Loc, \locinit, \Clocks, \Alphabet,$ $\Invariant, \Label, \Final, \Transitions)$, where: 
\begin{enumerate*}
    \item $\Loc$ is a finite set of locations,
    \item $\locinit\in \Loc$ is an initial location, 
    \item $\Clocks$ is a finite set of clocks, 
    \item $\Alphabet$ is a finite set of actions,
    \item $\Invariant$: $\Loc\to\Constraints(\Clocks)$ is a function that associates to each location clock guard in which $\compOp  \in  \{<,\le\}$ (called \emph{invariant}),
    \item $\Label: \Loc \to 2^{\Atomic}$ is a labeling function for the locations ($\Atomic$ is a set of atomic propositions), %\dm{check if needed}
    \item $\Final \subseteq \Loc$ is a set of final locations,
    \item $\Transitions \subseteq \Loc \times \Alphabet \times \Constraints(\Clocks) \times 2^{\Clocks} \times \Loc$ is a finite set of transitions.
\end{enumerate*}
\end{definition}
\noindent A transition  $e= (\loci{1}, \action, \guard, \resets, \loci{2}) \in \Transitions$ can be written as 
$\loci{1} \transitionArrow{a, \guard, \resets} \loci{2}$ where $\loci{1}$ and $\loci{2}$ are the source and target locations, 
$\action$ the action, 
$\guard$ a guard, 
$\resets$ the set of clocks to reset. To define the semantics of \iac{ta}, we recall the definition of \iac{tts}. A \ac{tts} is a tuple $(\states, \initialstate, \finalstates, \Alphabet, \Label_\tts, \transitions)$ where
\begin{enumerate*}
    \item $\states$ is a set of states, 
    \item $\initialstate$ is the initial state,
    \item $\finalstates$ is the set of final states,
    \item $\Alphabet$ is a set of actions,
    \item $\Label_\tts:$ $\states \to 2^{\Atomic}$ is a labeling function,
    \item $(\statein, \alpha, \statein') \in \transitions \subseteq \states \times (\Alphabet \cup  \mathbb{R}_{\ge 0}) \times \states$ is a labeled transition relation with two kinds of transitions: discrete (if $\alpha \in \Alphabet$) and delay (if $\alpha \in \mathbb{R}_{\ge 0}$) transitions.
\end{enumerate*}

\begin{definition}[Semantics of \iac{TA}]
    \label{def:defSemTA}
    Let $\TA = (\Loc, \locinit, \Clocks, \Alphabet, \Invariant, \Label, \Final, \Transitions)$ be \iac{TA}. 
    The semantics of $\TA$ is given by the \ac{tts} $\tts(\TA)  = (\states, \initialstate, \finalstates, \Alphabet,$ $ \Label_\tts, \transitions)$  where:
    \begin{enumerate*}
        \item $\states \subseteq \Loc \times \Rplus^{\Clocks}$,
        \item $\initialstate= (\locinit, \ClocksZero)$ with   $\ClocksZero \models \Invariant(\locinit)$, 
        \item $\finalstates \subseteq \Final \times \Rplus^{\Clocks}$,
        \item $\Label_\tts ((\loc, \clockval)) = \Label(\loc) \cup \set{\guard \in \Constraints(\Clocks) \mid \clockval \models \guard }$,
        \item $\transitions \subseteq \states \times (\Alphabet \cup \mathbb{R}_{\ge 0}) \times \states$ is a transition relation defined by the following two rules:
        \end{enumerate*}
        \begin{itemize}
            \item \textbf{Discrete transition: } $(\loci{1}, \clockval)\transitionArrow{a} (\loci{2}, \clockval')$
            for $\action\in \Alphabet$ \iff{} $\loci{1} \transitionArrow{\action, \guard, \resets} \loci{2}, \clockval \models \guard \land \clockval' = \reset{\clockval}{\resets}$ and $\clockval' \models \Invariant(\loci{2})$ and, 
            \item \textbf{Delay transition: } $(\loci{1}, \clockval)\transitionArrow{d} (\loci{1}, \clockval+d)$, for some  $d\in \Rplus$ iff $\clockval + d \models \Invariant(\loci{1})$. 
        \end{itemize}
    
\end{definition}

\noindent In the remainder of the paper, we consider \iac{ta}~$\TA$ and its associated \ac{tts}~$\tts$. A finite run of $\TA$ is a finite path of $\tts$ in which delay and discrete transitions alternate: 
$\run  = \initialstate \transitionArrow{d_0} \statei{1}' \transitionArrow{\actioni{0}} \statei{1} \transitionArrow{d_1} \statei{2}' \transitionArrow{\actioni{1}} \statei{2} \ldots \statei{n-1} \transitionArrow{d_{n-1}} \statei{n}'  \transitionArrow{\actioni{n-1}} \statei{n}$ 
or more compactly  $\run = \statei{0} \transitionArrow{d_0, \actioni{0}} \statei{1} \transitionArrow{d_1, \actioni{1}} \statei{2} \transitionArrow{d_2, \actioni{2}} \statei{3} \ldots \statei{n-1} \transitionArrow{d_{n-1}, a_{n-1}} \statei{n}$. 

\noindent The last state $\statei{n}$ of a run $\run$ is denoted by $\last(\run)$.
The \emph{duration} of a run is the sum of the delays along the run: $\dur(\rho)=\sum_{i=0}^{n-1} d_i$.
Two runs are \emph{duration-equivalent} if they have the same duration: we write $\run \equivt \run'$ iff $\dur(\run)=\dur(\run')$. We denote by $\Paths_{\tts}(\statein)^{+}$ the set of all paths of $\tts$ that start from $\statein$ and end in a final state.

%------------------------------------------------------------
\newcommand{\sclock}{x}
\newcommand{\sclocky}{y}
\newcommand{\assign}{\ensuremath{\leftarrow}}
\newcommand{\styleact}[1]{\textcolor{green!50!black}{#1}}
\definecolor{colordisc}{rgb}{1, 0, 1}
\newcommand{\styledisc}[1]{\textcolor{colordisc}{#1}}
% Macros for ATM
\newcommand{\ATMaskPassword}{\ensuremath{\styleact{askPwd}}}
\newcommand{\ATMcorrectAmount}{\ensuremath{\styleact{correctAmount}}}
\newcommand{\ATMcorrectPassword}{\ensuremath{\styleact{correctPwd}}}
\newcommand{\ATMrequestBalance}{\ensuremath{\styleact{reqBalance}}}
\newcommand{\ATMfinish}{\ensuremath{\styleact{finish}}}
\newcommand{\ATMincorrectAmount}{\ensuremath{\styleact{incorrectAmount}}}
\newcommand{\ATMincorrectPassword}{\ensuremath{\styleact{incorrectPwd}}}
\newcommand{\ATMnormalWithdrawal}{\ensuremath{\styleact{normalWithdraw}}}
\newcommand{\ATMpressFinish}{\ensuremath{\styleact{pressFinish}}}
\newcommand{\ATMpressOK}{\ensuremath{\styleact{pressOK}}}
\newcommand{\ATMquickWithdrawal}{\ensuremath{\styleact{quickWithdraw}}}
\newcommand{\ATMrestart}{\ensuremath{\styleact{restart}}}
\newcommand{\ATMstart}{\ensuremath{\styleact{start}}}
\newcommand{\ATMtakeCash}{\ensuremath{\styleact{\textcolor{red}{takeCash}}}}
\newcommand{\ATMnbFP}{\ensuremath{\styledisc{nbFP}}}
\newcommand{\ATMI}{\ensuremath{I}}
\newcommand{\ATMW}{\ensuremath{W}}
\newcommand{\ATMWP}{\ensuremath{WP}}
\newcommand{\ATMWC}{\ensuremath{WC}}
\newcommand{\ATMWA}{\ensuremath{WA}}
\newcommand{\ATMPNW}{\ensuremath{PNW}}
\newcommand{\ATMPQW}{\ensuremath{PQW}}
\newcommand{\ATMDB}{\ensuremath{DB}}
\newcommand{\ATMMAN}{\ensuremath{MAN}}
\newcommand{\ATMMAQ}{\ensuremath{MAQ}}
\newcommand{\ATMOO}{\ensuremath{OO}}
\newcommand{\ATMT}{\ensuremath{T}}
\newcommand{\ATMC}{\ensuremath{C}}
\newcommand{\ATME}{\ensuremath{E}}
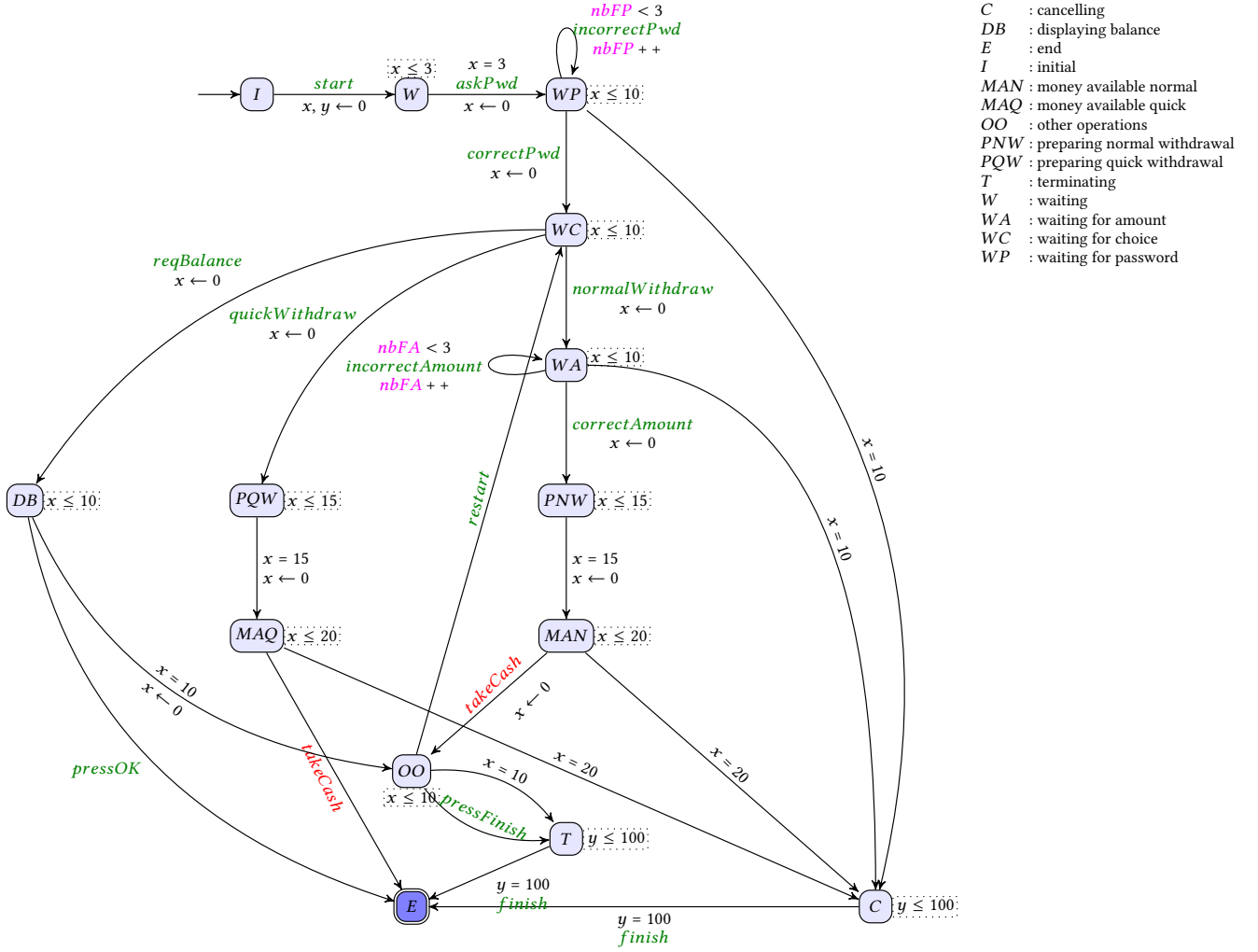
\begin{figure*}[tb]
    \tikzstyle{every node}=[initial text=]
    \tikzstyle{location}=[rectangle, rounded corners, minimum size=12pt, draw=black, fill=blue!10, inner sep=2pt]
    \tikzstyle{invariant}=[draw=black, dotted, inner sep=1pt, node distance=0] % xshift=1em,
    \tikzstyle{final}=[double, fill=blue!50]
    \centering
	\scriptsize
	
	\resizebox{\textwidth}{!}{
		\begin{tikzpicture}[auto, ->, >=stealth', scale=1.0, xscale=2, yscale=1.75]
			
			\node[location, initial] at (-2, -2) (initial) {$\ATMI$};
			
			\node[location] at (-1, -2) (welcome) {$\ATMW$};
			\node[invariant, above=of welcome] {$\sclock \leq 3$};
			
			\node[location] at (0, -2) (waitingPassword) {$\ATMWP$};
			\node[invariant, right=of waitingPassword] {$\sclock \leq 10$};
			
			\node[location] at (0, -3) (waitChoice) {$\ATMWC$};
			\node[invariant, right=of waitChoice] {$\sclock \leq 10$};
			
			\node[location] at (0, -4) (waitingAmount) {$\ATMWA$};
			\node[invariant, right=of waitingAmount, yshift=.5em] {$\sclock \leq 10$};
			
			\node[location] at (0, -5) (preparingWithdrawalNormal) {$\ATMPNW$};
			\node[invariant, right=of preparingWithdrawalNormal] {$\sclock \leq 15$};
			
			\node[location] at (-2, -5) (preparingWithdrawalQuick) {$\ATMPQW$};
			\node[invariant, right=of preparingWithdrawalQuick] {$\sclock \leq 15$};
			
			\node[location] at (-3.5, -5) (displayingBalance) {$\ATMDB$};
			\node[invariant, right=of displayingBalance] {$\sclock \leq 10$};
			
			\node[location] at (0, -6) (moneyAvailableNormal) {$\ATMMAN$}; % private
			\node[invariant, right=of moneyAvailableNormal] {$\sclock \leq 20$};
			
			\node[location] at (-2, -6) (moneyAvailableQuick) {$\ATMMAQ$}; % private
			\node[invariant, right=of moneyAvailableQuick] {$\sclock \leq 20$};
			
			\node[location] at (-1, -7) (otherOperation) {$\ATMOO$};
			\node[invariant, below=of otherOperation] {$\sclock \leq 10$};
			
			\node[location] at (0, -7.5) (terminating) {$\ATMT$};
			\node[invariant, right=of terminating] {$\sclocky \leq 100$};
			
			\node[location] at (2, -8) (cancelling) {$\ATMC$};
			\node[invariant, right=of cancelling] {$\sclocky \leq 100$};
			
			\node[location, final] at (-1, -8) (theEnd) {$\ATME$};

			\path (initial) edge[] node[above]{$\ATMstart$} node[below]{$\sclock, \sclocky \assign 0$} (welcome);
			
			\path (welcome) edge[] node[above, align=center]{$\sclock = 3$\\$\ATMaskPassword$} node[below]{$\sclock \assign 0$} (waitingPassword);
			
			\path (waitingPassword) edge[loop above] node[right, align=center]{$\ATMnbFP < 3$\\$\ATMincorrectPassword$\\$\ATMnbFP++$} (waitingPassword);
			
			\path (waitingPassword) edge[] node[left, align=center]{$\ATMcorrectPassword$\\$\sclock \assign 0$} (waitChoice);
			
			\path (waitingPassword) edge[bend left] node[sloped, align=center]{$\sclock = 10$} (cancelling);
			
			\path (waitChoice) edge[bend right] node[left, align=center]{$\ATMquickWithdrawal$\\$\sclock \assign 0$} (preparingWithdrawalQuick);
			
			\path (waitChoice) edge[] node[align=center]{$\ATMnormalWithdrawal$\\$\sclock \assign 0$} (waitingAmount);
			
			\path (waitChoice) edge[bend right] node[above left, align=center]{$\ATMrequestBalance$\\$\sclock \assign 0$} (displayingBalance);
			
			\path (waitingAmount) edge[loop left] node[align=center]{$\styledisc{nbFA} < 3$\\$\ATMincorrectAmount$\\$\styledisc{nbFA}++$} (waitingAmount);
			
			\path (waitingAmount) edge[] node[align=center]{$\ATMcorrectAmount$\\$\sclock \assign 0$} (preparingWithdrawalNormal);
			
			\path (waitingAmount) edge[out=0, in=90] node[sloped, align=center]{$\sclock = 10$} (cancelling);
			
			\path (preparingWithdrawalNormal) edge[] node[align=center]{$\sclock = 15$\\$\sclock \assign 0$} (moneyAvailableNormal);
			
			\path (preparingWithdrawalQuick) edge[] node[align=center]{$\sclock = 15$\\$\sclock \assign 0$} (moneyAvailableQuick);
			
			\path (moneyAvailableNormal) edge[] node[sloped, align=center]{$\sclock = 20$} (cancelling);
			
			\path (moneyAvailableNormal) edge[] node[sloped, above, xshift=1em]{$\ATMtakeCash$} node[sloped, below, xshift=2em,yshift=-1em]{$\sclock \assign 0$} (otherOperation);
			
			\path (moneyAvailableQuick) edge[] node[above, sloped, align=center]{$\sclock = 20$} (cancelling);
			
			\path (moneyAvailableQuick) edge[] node[sloped, below, align=center]{$\ATMtakeCash$} (theEnd);
			
			\path (displayingBalance) edge[bend right] node[sloped, above]{$\sclock = 10$} node[sloped, below]{$\sclock \assign 0$} (otherOperation);
			
			\path (displayingBalance) edge[bend right] node[below left, align=center]{$\ATMpressOK$} (theEnd);
			
			\path (otherOperation) edge[bend left] node[sloped,align=center]{$\sclock = 10$} (terminating);
			\path (otherOperation) edge[bend right] node[sloped, %above=0.3cm
			]{$\ATMpressFinish$} (terminating);
			
			\path (otherOperation) edge[] node[sloped, align=center]{$\ATMrestart$} (waitChoice);
			
			\path (cancelling) edge[] node[align=center]{$\sclocky = 100$\\$\ATMfinish$} (theEnd);
			
			\path (terminating) edge[] node[align=center]{$\sclocky = 100$\\$\ATMfinish$} (theEnd);

			% LEGEND
			\node[] at (3.5, -2.3) (legend) {
				\begin{tabular}{l @{~:~} l}
					$C$ & cancelling\\
					$DB$ & displaying balance\\
					$E$ & end\\
					$I$ & initial\\
					$MAN$ & money available normal\\
					$MAQ$ & money available quick\\
					$OO$ & other operations\\
					$PNW$ & preparing normal withdrawal\\
					$PQW$ & preparing quick withdrawal\\
					$T$ & terminating\\
					$W$ & waiting\\
					$WA$ & waiting for amount\\
					$WC$ & waiting for choice\\
					$WP$ & waiting for password\\
				\end{tabular}
			};
		\end{tikzpicture}
	}
 %\vspace{-4pt}
\caption{Running example: ATM benchmark \cite{ABLM22}}
	\label{fig:ATM}
\end{figure*}
%-----------------------------------------------------------

%-----------------------------------------------------------
\begin{example}
%As a running example and case study for our model and logic, we consider \iac{TA} model of an ATM system (\cref{fig:ATM}). \ATMWP{} is a location, where we can stay as long as clock $\clock \leq 10$. Before $\clock >10$, we must take one of the two outgoing transitions. If $\clock = 10$ we can go to the $\ATMC$ location. If we take the $\ATMcorrectPassword$ action, we go to the \ATMWC{} location, while resetting clock $\clock$. Note that this model uses a discrete variable $\ATMnbFP$ to count password failures; since it ranges over a finite domain, it is a syntactic sugar (and this model is still \iac{ta}).
As a running example and case study, we consider \iac{TA} model of an ATM system (\cref{fig:ATM}). In location \ATMWP{}, the system may remain while $\clock\leq 10$; before this bound is exceeded, one of the outgoing transitions must be taken. When $\clock=10$, the system may move to location $\ATMC$. Taking the action $\ATMcorrectPassword$ leads to \ATMWC{} and resets $\clock$. The model uses the finite-domain variable $\ATMnbFP$ to count password failures, this is only syntactic sugar, so the model remains \iac{ta}.
\end{example}
%-----------------------------------------------------------
%%%%%%%%%%%%%%%%%%%%%%%%%%%%%%%%%%%%%%%%%%%%%%%%%%%%%%%%%%%%%%
\subsection{Zone Graph and Predecessor Operator}
\label{subsection:graph}
%%%%%%%%%%%%%%%%%%%%%%%%%%%%%%%%%%%%%%%%%%%%%%%%%%%%%%%%%%%%%%

%Since $\tts$ has infinitely many states due to the continuous nature of clocks, its semantics cannot be represented directly by a finite-state transition system.  To obtain a finite symbolic abstraction, we use the \ac{zonegraph} abstraction \cite{HNSY94,BYJ04} (formally defined in \cref{def:zonegraph}), which is the standard symbolic representation underlying most verification procedures \acp{TA}. Symbolic verification tools for real-time systems typically manipulate sets of clock valuations in the form of zones, in practice, these zones are efficiently represented by \acp{DBM}. Given a set of clocks $\Clocks$, a \emph{zone} is a set of clock valuations satisfying a conjunction of constraints of the form $\clock \compOp d$ or $\clock - \clocky \compOp c$, where $\clock, \clocky \in \Clocks$, $c \in \Z$, and $\compOp \in \{<, \leq, =, \geq, >\}$. Given a zone $\zone$ and a delay $d \in \Rplus$, we denote by $\timeelpaseof{\zone}{d}$ the zone $\Set{\clockval + d}{\clockval \in \zone}$, which represents the evolution of clock valuations after a delay of $d$ time units and $\timeelpase{\zone} = \bigcup_{d \in \Rplus} \timeelpaseof{\zone}{d}$. A \emph{symbolic state} is a pair $\symZone = (\loc, \zone)$ where $\loc \in \Loc$ is a location and $\zone$ is a zone over $\Clocks$. A symbolic state $(\loc, \zone)$ represents the set of concrete states $\{(\loc, \clockval) \mid \clockval \in \zone\}$.

Since $\tts$ has infinitely many states due to dense time, its semantics cannot be handled directly as a finite-state transition system. We therefore use the \ac{zonegraph} abstraction~\cite{HNSY94,BYJ04}, formally defined in \cref{def:zonegraph}, which is the standard symbolic representation for verification of \acp{TA}. It manipulates sets of clock valuations, called \emph{zones}, usually represented efficiently by \acp{DBM}. Given a set of clocks $\Clocks$, a zone is a set of valuations satisfying a conjunction of constraints of the form $\clock\compOp d$ or $\clock-\clocky\compOp c$, where $\clock,\clocky\in\Clocks$, $c\in\Z$, and $\compOp\in\{<,\leq,=,\geq,>\}$. For a zone $\zone$ and a delay $d\in\Rplus$, we write $\timeelpaseof{\zone}{d}=\Set{\clockval+d}{\clockval\in\zone}$, and $\timeelpase{\zone}=\bigcup_{d\in\Rplus}\timeelpaseof{\zone}{d}$. A \emph{symbolic state} is a pair $\symZone=(\loc,\zone)$ with $\loc\in\Loc$ and $\zone$ a zone over $\Clocks$, representing the concrete states $\{(\loc,\clockval)\mid \clockval\in\zone\}$.

\newcommand{\SymStates}{\ensuremath{\mathcal{S}}}
\newcommand{\initialSymState}{\ensuremath{s_0}}
\newcommand{\SetSymStates}{\ensuremath{\mathcal{P}}}
\newcommand{\SymState}{\ensuremath{s}}
\newcommand{\SymTransitions}{\ensuremath{\rightarrow}}

\begin{definition}[\acl{zonegraph}]
\label{def:zonegraph}
The \ac{zonegraph} of $\TA$ is a transition system $\zonegraph = (\SymStates, \initialSymState, (\Alphabet \cup \set{\epsilon}), \Label_\ttz, \SymTransitions)$ where:
\begin{enumerate*}
    \item $\SymStates$ is the least set of symbolic states $(\loc,\zone)$ containing $\initialSymState$ and closed under the transition rules below.
    % is a finite set of symbolic states $(\loc,\zone)$ such that every valuation in $\zone$ satisfies the invariant of $\loc$, i.e. $\forall \clockval\in\zone,\ \clockval\models \Invariant(\loc)$,
    \item $\initialSymState = (\locinit, \{\ClocksZero\})$ is the initial symbolic state,
    \item $\Alphabet$ is the set of action labels of $\TA$ and $\epsilon\notin\Alphabet$ is a special symbol used to denote delay transitions in the zone graph  semantics,
    \item $\Label_\ttz ((\loc, \zone)) = \Label(\loc)\ \cup\
\{\guard \in \Constraints(\Clocks)\mid \forall \clockval \in \zone,\ \clockval \models \guard\}$,
    \item $\SymTransitions \subseteq \SymStates \times (\Alphabet \cup \set{\epsilon}) \times \SymStates$ is the transition relation defined by: 
    \end{enumerate*}
    \begin{itemize}
    
    \item $(\loc, \zone) \transitionArrow{a} (\loc', \zone')$ if there exists a transition $\loc \transitionArrow{a, \guard, \resets} \loc'$ in $\TA$ such that $\zone' = \reset{(\zone \cap \llbracket \guard\rrbracket)}{\resets} \cap \llbracket \Invariant(\loc') \rrbracket$,
    \item $(\loc, \zone) \transitionArrow{\epsilon} (\loc, \zone')$  if $\zone' = \timeelpase{\zone} \cap \llbracket \Invariant(\loc) \rrbracket$.
    \end{itemize}
       % \dm{Vérifier quelles sont les étiquettes des transitions. $\Alphabet$ ne semble pas correct?}
%\reviewer{2}{Reviewer 2 remarks that an element ``Lab'' appears to be missing in this definition (page 9, line -9 in the submitted version). Verify that (a) the labeling function $Label_ttz$ and (b) the set of final symbolic states $SymFinal$ are properly included in the zone graph tuple and its enumeration.}
\end{definition}

\noindent Whenever a finite symbolic abstraction is required, the ordinary automaton clocks are normalized using a standard sound zone abstraction. Formula satisfaction, however, is defined on concrete extended states $(\ell,\nu)$ rather than directly on zones. Indeed, a zone may contain valuations that satisfy a formula and valuations
that do not. Therefore, during symbolic model checking, the set of concrete states satisfying a formula is represented as a union of zones. When a zone is not homogeneous with respect to the formula being checked, it is refined into smaller zones.  Using standard normalization and abstraction operations on zones, the symbolic state space of $\zonegraph(\TA)$ is finite~\cite{HNSY94,BYJ04}. Thus, $\zonegraph(\TA)$ is a finite abstraction of the infinite timed transition system $\tts(\TA)$, preserving the information needed for symbolic verification. In the sequel, we work on $\zonegraph(\TA)$. Given $\SetSymStates\subseteq\SymStates$, the predecessor set of $\SetSymStates$ is $\Pre(\SetSymStates) = \Set{\symZone'\in\SymStates}{ \exists \symZone\in\SetSymStates,\ \exists \alpha\in(\Alphabet\cup\{\epsilon\}) \text{ such that } \symZone'\transitionArrow{\alpha}\symZone }$. Its computation combines a backward delay step $\Pre^c(\SetSymStates)$ and a backward discrete step $\Pre^d(\SetSymStates)$~\cite{HNSY94}. Dually, the successor set of $\SetSymStates$ is $\Post(\SetSymStates) = \Set{\symZone'\in\SymStates}{ \exists \symZone\in\SetSymStates,\ \exists \alpha\in(\Alphabet\cup\{\epsilon\}) \text{ such that } \symZone\transitionArrow{\alpha}\symZone' }$.

\begin{lemma}
\label{lem:pre-correctness}
Let $\setX\subseteq\SymStates$ and let $\symZone=(\loc,\zone)$.
Then $\symZone\in\Pre(\setX)$ iff there exist $\clockval\in\zone$, a concrete successor $(\loc',\clockval')$, and
$\symZone'=(\loc',\zone')\in\setX$ such that $(\loc,\clockval)\longrightarrow(\loc',\clockval')$ and $\clockval'\in\zone'$.
\end{lemma}

\begin{proof}
By definition, $\Pre(\setX)$ contains precisely the symbolic states having a symbolic successor in $\setX$. For the forward direction, let $\symZone\in\Pre(\setX)$. Then there exists $\symZone'\in\setX$ such that $\symZone\Rightarrow\symZone'$. By soundness of the symbolic successor construction, this transition is
witnessed by concrete valuations $\clockval\in\zone$ and $\clockval'\in\zone'$ satisfying $(\loc,\clockval)\longrightarrow(\loc',\clockval')$. Conversely, suppose such concrete valuations exist. By completeness of the
symbolic successor construction, the corresponding concrete transition is represented by a symbolic transition $\symZone\Rightarrow\symZone'$. Since $\symZone'\in\setX$, it follows that $\symZone\in\Pre(\setX)$.
\end{proof}

%Given a symbolic state $\symZone\in\SymStates$, we define its set of predecessors by
%$\Pre(\symZone)
%=
%\Set{\symZone' \in \SymStates}{
%\exists \alpha\in(\Alphabet\cup\{\epsilon\}) \text{ such that } \symZone' %\transitionArrow{\alpha} \symZone
%}$.

%%%%%%%%%%%%%%%%%%%%%%%%%%%%%%%%%%%%%%%%%%%%%%%%%%%%%%%%%%%%%%%%%%%%%%
%%%%%%%%%%%%%%%%%%%%%%%%%%%%%%%%%%%%%%%%%%%%%%%%%%%%%%%%%%%%%%%%%%%%%%
\section{ET-Opacity Logic} \label{sec:etol}
%%%%%%%%%%%%%%%%%%%%%%%%%%%%%%%%%%%%%%%%%%%%%%%%%%%%%%%%%%%%%%%%%%%%%%
%%%%%%%%%%%%%%%%%%%%%%%%%%%%%%%%%%%%%%%%%%%%%%%%%%%%%%%%%%%%%%%%%%%%%%

%%%%%%%%%%%%%%%%%%%%%%%%%%%%%%%%%%%%%%%%%%%%%%%%%%%%%%%%%%%%
%\subsection{Execution-Time Opacity (ET-opacity)}
%%%%%%%%%%%%%%%%%%%%%%%%%%%%%%%%%%%%%%%%%%%%%%%%%%%%%%%%%%%%%

%Contrary to \cite{ADJ22,ALLMS23}, where the secret is defined as the visit of a particular control location, we allow here an arbitrary private property over complete runs. In the following, we assume a secret property defined by a predicate $\secretProp$.
%Unlike~\cite{ADJ22,ALLMS23}, where secrecy is tied to visiting a specific location, we allow secrets to be arbitrary private properties over complete runs, defined by a predicate $\secretProp$.

Unlike~\cite{ADJ22,ALLMS23}, where secrecy is location-based, we define secrets as private state predicates evaluated at the endpoint of terminig runs. % run properties given by a predicate $\secretProp$.

\begin{definition}[\acl{etopacity}]
\begin{enumerate}
 \item $\tts$ satisfies \emph{existential \ac{etopacity}} at $\statein$ with respect to $\secretProp$, denoted
        $\tts, \statein \models \exists\text{-ETO}(\secretProp)$
        iff there exist $\run,\run' \in \Paths_{\tts}(\statein)^+, 
        \run \equivt \run'
        \land
        (\tts, \last(\run) \models \secretProp)
        \land
        (\tts, \last(\run') \models \neg \secretProp.)
        $

    \item 
  %  $ \tts, \statein \models \forall\text{-ETO}(\secretProp) $ if  for all possible durations of runs from $\statein$, there exist two duration-equivalent runs that differ on the satisfaction of $\secretProp$, \ie{} 
  %$
   % \forall \run \in \Paths_{\tts}(\statein)^+, \exists \run' \in \Paths_{\tts}(\statein)^+,
   % \run \equivt \run'
   % \land
   % (\tts, \last(\run) \models \secretProp)
   % \land
   % (\tts, \last(\run') \models \neg \secretProp.)
   %$
    $\tts$ satisfies \emph{universal \ac{etopacity}} at $\statein$ with respect to $\secretProp$, denoted
    $\tts,\statein\models\forall\text{-ETO}(\secretProp)$, iff every terminating run has a duration-equivalent run with the opposite truth value of the private predicate:
    $
      \forall \run\in\Paths_{\tts}(\statein)^+\;\exists\run'\in\Paths_{\tts}(\statein)^+:
      \run \equivt \run'
      \ \land\
      $
      $\Bigl(
      \tts,\last(\run)\models\secretProp
      \Leftrightarrow
      \tts,\last(\run')\models\neg\secretProp
      \Bigr).
    $
\end{enumerate}
\noindent The second clause is symmetric and is equivalent to equality between the sets of durations of private and non-private runs. If only one-sided (weak) opacity is desired, the last equivalence can instead be replaced by :
$\tts,\last(\run)\models\secretProp\Rightarrow
 \tts,\last(\run')\models\neg\secretProp$.
\end{definition}

%%%%%%%%%%%%%%%%%%%%%%%%%%%%%%%%%%%%%%%%%%%%%%%%%%%%%%%%%%%%
\subsection{ET-Opacity Logic (\ETOL)}
%%%%%%%%%%%%%%%%%%%%%%%%%%%%%%%%%%%%%%%%%%%%%%%%%%%%%%%%%%%%%

%In this section, we introduce the syntax and semantics of \ETOL, which is designed to capture \ac{etopacity}. Intuitively, a formula of the form $\ETexists\prop$ expresses that there exists at least one execution duration performed by two executions, one satisfying $\prop$ and the other violating $\prop$. In other words, there exists a duration. Dually, a formula of the form $\ETforall\prop$ expresses that this indistinguishability holds for all possible execution times.

This section introduces the syntax and semantics of \ETOL, designed to capture \ac{etopacity}. Intuitively, $\ETexists\prop$ states that there exists an execution duration shared by two executions, one satisfying $\prop$ and one violating it. Dually, $\ETforall\prop$ requires such indistinguishability for all possible execution times.

\begin{definition}
Let $\Atomic$ be a set of atomic propositions, $\Clocks$ the set of automaton clocks, and $\ClocksFormula$ a non-empty set of formula clocks such that  $\Clocks \cap \ClocksFormula = \emptyset$.  Formulas of \ETOL are defined by the following (minimal) grammar:
$$
\begin{aligned}
 \prop ::= 
    &   \top \mid 
        \atomic \mid 
        \neg \prop \mid 
        \guard  \mid 
        \prop \land \prop  \mid
        \clockformula.\prop \mid 
        \ETexists (\prop \until \prop) \mid \\
    &    \ETexists (\prop \release \prop) \mid 
         \ETforall (\prop \until \prop) \mid 
        \ETforall (\prop \release \prop)
\end{aligned}
$$
\noindent where 
    $\atomic \in \Atomic$, 
    $\clockformula \in \ClocksFormula$ and $\guard \in \Constraints(\Clocks \cup \ClocksFormula)$.
\end{definition}
%\dm{Confirmer si $\prop_1, \prop_2 \in \ETOL$. On peut mettre un $ETexists$ dans un $ETexists$ ? un $AG$ dans un $ETexists$ ?}

%\noindent It is possible to compare formula clocks $\ClocksFormula{}$ and automaton clocks $\Clocks{}$ by means of clock constraints of the form $\guard$, which may involve both clocks from the \ac{ta} and clocks introduced by the formula. The Boolean connectives $\bot$, $\vee$, and $\to$ are defined as usual. A clock $\clockformula$ occurring in a formula of the form $\clockformula.\prop$ is called a \emph{freeze identifier} and binds the formula clock $\clockformula$ within $\prop$. Intuitively, the formula $\clockformula.\prop$ holds in a state $\state$ if $\prop$ holds in the extended state obtained by resetting the value of clock $\clockformula$ to~0 in $\state$. Freeze identifiers can be used in conjunction with temporal operators to express common real-time requirements, such as punctuality constraints, bounded response times, or deadlines. \ETOL introduces two dedicated \ac{etopacity} operators: the \emph{existential \ac{etopacity} operator} $\ETexists{}$  corresponding to $\exists$-\ac{etopacity}, and the \emph{universal \ac{etopacity} operator} $\ETforall$ corresponding to $\forall$-\ac{etopacity}. For each \ac{etopacity} operator in $\set{\ETexists , \ETforall }$, derived temporal operators are defined similarly:
\noindent Clock constraints $\guard$ may compare automaton clocks in $\Clocks$ with formula clocks in $\ClocksFormula$. The Boolean connectives $\bot$, $\vee$, and $\to$ are defined as usual. In a formula $\clockformula.\prop$, the clock $\clockformula$ is a \emph{freeze identifier}: it binds $\clockformula$ in $\prop$ and resets it to $0$ before evaluating $\prop$. It is a logical clock and does not modify any clock of the \ac{ta}. \ETOL introduces two \ac{etopacity} operators: the existential operator $\ETexists$, for $\exists$-\ac{etopacity}, and the universal operator $\ETforall$, for $\forall$-\ac{etopacity}. For each operator in $\set{\ETexists,\ETforall}$, the derived temporal operators are defined similarly:
$$
\begin{array}{r l p{.01\linewidth} r l}
\ETexists \eventualy \prop :=& \ETexists (\top \until \prop) & & \ETforall \eventualy \prop :=& \ETforall (\top \until \prop)  \\
\ETexists \globaly \prop :=& \ETexists (\bot \release \prop) & & \ETforall \globaly \prop :=& \ETforall (\bot \release \prop)  \\
\ETexists (\prop \mathsf{W} \propOther)
:=& \ETexists (\propOther \release (\prop \vee \propOther)) & & \ETforall (\prop \mathsf{W} \propOther)
:=& \ETforall (\propOther \release (\prop \vee \propOther)).
\end{array}
$$
The size $|\prop|$ of a formula $\prop$ is defined as the number of logical and temporal connectives occurring in $\prop$. %\dm{utile?} 
\begin{example}
    Consider the formula $\prop = \clockformula.\ETexists\big( \prop_1 \until (\prop_2 \wedge \clockformula \leq 7) \big)$.
    The binder resets $\clockformula$ before the opacity modality is evaluated. Hence $\clockformula$ measures the time elapsed from the current state, and the right-hand side of the until must be reached within $7$ time units on the run satisfying the path formula. %Existential opacity additionally requires a duration-equivalent run on which the same until formula is false.
    % $\prop = \ETexists\big(\clockformula.(\prop_1 \wedge \clockformula \leq 7) \until \prop_2\big)$.
   % The freeze operator $\clockformula.(\cdot)$ resets $\clockformula$ before the until formula is evaluated. Thus, $\clockformula$ measures the time until $\prop_2$ is reached, and the formula requires $\prop_1$ to hold until $\prop_2$ becomes true within $7$ time units. The opacity operator then states that this timed behaviour cannot be distinguished, by observing only total execution time, from behaviours where the property does not hold.
%The freeze operator $\clockformula.(\cdot)$ resets $\clockformula$ before the until formula is evaluated, so $\clockformula$ measures the time elapsed until $\prop_2$ is reached. Hence, the formula requires $\prop_1$ to hold until $\prop_2$ becomes true within $7$ time units. The opacity operator further requires this timed behaviour to be indistinguishable, in terms of total execution time, from behaviours where the property does not hold.
    %This formula intuitively states that for every execution satisfying $\clockformula.\big(\prop_1 \wedge \clockformula \leq 7\big) \until \prop_2$ (\ie{} where $\prop_1$ is true until $\prop_2$ becomes true before 7 time units), there exists another execution with exactly the same total execution time measured by clock $\clockformula$ for which this property does not hold. Consequently, an observer who can only observe execution times cannot determine whether the property $\clockformula.\big(\prop_1 \wedge \clockformula \leq 7\big) \until \prop_2$ was satisfied. 
\end{example}
\begin{example}
%\reviewer{2}{The formula in this example may be trivially satisfied in part: since $\clockformula_1$ is reset to $0$ by the freeze operator $\clockformula_1.(-)$, the constraint $\clockformula_1 \leq 6$ holds trivially at the outermost level ($\clockformula_1 = 0 \leq 6$). A temporal operator is likely missing between $\clockformula_1.(-)$ and the inner conjunct $\clockformula_1 \leq 6$, so that the clock $\clockformula_1$ accumulates some duration before being compared. Revise the formula accordingly.}
Distinct timing references can be combined without nesting opacity modalities. For instance, %$\prop= \clockformula_1.\Bigl( (\clockformula_1\leq 6)  \land \clockformula_2.\ETexists\bigl( \prop_2\until(\prop_3\land\clockformula_2\leq4) \bigr) \Bigr)$.
$\prop =
\clockformula_1.(
  \ETexists\bigl(
     \prop_1 \until
     (\prop_2 \land \clockformula_1 \leq 6)
  \bigr)
  \land
  \clockformula_2.\ETexists\bigl(
     \prop_2 \until
     (\prop_3 \land \clockformula_2 \leq 4)
  \bigr)).$
The two freeze clocks introduce independent temporal reference points. The clock $\clockformula_1$ is reset before evaluating the first opacity formula and therefore measures the time elapsed until $\prop_2$ is reached, which is required to occur within $6$ time units. Independently, $\clockformula_2$ is reset before evaluating the second opacity formula and measures the time elapsed until $\prop_3$ is reached, which must occur within $4$ time units. Thus, neither clock constraint is evaluated immediately after its reset, and both impose non-trivial timing requirements.
\end{example} 
\begin{example}
Automaton and formula clocks may occur in the same local constraints: $\prop=\clockformula.\ETexists\bigl(
(\prop_1\land x\leq2)\until
(\prop_2\land\clockformula\leq5\land x\geq1)
\bigr)$. The automaton clock $x$ follows the resets of the model, whereas $\clockformula$ measures time from the logical freeze point.
%Consider the formula $\prop=\clockformula.\ETexists\big((\prop_1 \wedge x \leq 2)\until (\prop_2 \wedge \clockformula \leq 5 \wedge x \geq 1)\big)$. %Here, $x$ is an automaton clock, while $\clockformula$ is a formula clock. The clock $x$ evolves with the system and may be reset by automaton transitions. By contrast, $\clockformula$ is reset locally by the freeze operator $\clockformula.(\cdot)$ and measures the elapsed time from that reset point. Thus, the constraints on $x$ describe the system timing, whereas $\clockformula\leq 5$ bounds the time to reach $\prop_2$.
%Here, $x$ is an automaton clock, while $\clockformula$ is a formula clock. The clock $x$ follows the system evolution and may be reset by transitions, whereas $\clockformula$ is reset locally by the freeze operator. Thus, constraints on $x$ describe the automaton timing, while $\clockformula\leq 5$ bounds the time elapsed from the freeze point until $\prop_2$ is reached.
%the automaton clock $x$ and the formula clock $\clockformula$ play different roles.  The constraint on $x$ refers to the timing behaviour of the system itself, whereas the constraint on $\clockformula$ measures time from the local reset introduced by the formula. %Thus, the formula combines model-level timing information and formula-level timing information in the same opacity specification.
\end{example}
%\noindent Unlike CTL, but for the same reasons as \ac{tctl}, \ETOL does not include a next-time operator. Indeed, in our context, the next operator would not be meaningful since time is continuous and there is no unique next moment in time.
\noindent As in \ac{tctl}, \ETOL does not include a next-time operator: over continuous time, there is no unique next moment.
\begin{definition}[Path semantics of \texorpdfstring{$\until$}{until} and \texorpdfstring{$\release$}{release}]
Let $\run$ be a finite run in $\Paths_{\tts}(\statein)^+$. We write :
\begin{itemize}
    \item $\run \models (\prop_1 \until \prop_2)$, iff there exists $i\in\{0,\ldots,n\}$ such that $\tts,\statei{i}\models \prop_2  \ \text{and} \  \tts,\statei{j}\models \prop_1$ for every $j< i$,
    \item  $\run \models (\prop_1 \release \prop_2)$, iff either
    $\forall i\in\{0,\ldots,n\},\ \tts,\statei{i}\models \prop_2,$
    or there exists $i\in\{0,\ldots,n\}$ such that $\tts,\statei{i}\models \prop_1 \wedge \prop_2 \ \text{and} \ \tts,\statei{j}\models \prop_2$  for every $j < i$.
\end{itemize}
\end{definition}

%\noindent Formulas of \ETOL are interpreted over \ac{TTS}. We are now ready to formally define the semantics of \ETOL formulas. An extended state $\state_{\clockvalformula}$ over $\states$ is a triple $(\loc, \clockval, \clockvalformula)$, where $\state = (\loc, \clockval)\in\states$ (in $\tts$) and $\clockvalformula$ a valuation for the formula clocks in $\ClocksFormula$. To facilitate reading, from this point onward we will use only the symbol $\state$ for an extended state.
%\noindent \ETOL formulas are interpreted over \acp{TTS}. An extended state is a triple $\statein_{\clockvalformula}=(\loc,\clockval,\clockvalformula)$, where $(\loc,\clockval)\in\states$ is a state of $\tts$ and $\clockvalformula$ is a valuation of the formula clocks in $\ClocksFormula$. For readability, we write simply $\statein$ for extended states in the sequel.
\noindent \ETOL formulas are interpreted over extended states
$\statein=(\loc,\clockval,\clockvalformula)$, where $(\loc,\clockval)$ is a concrete state of $\tts$ and
$\clockvalformula:\ClocksFormula\to\Rplus$ is a valuation of formula clocks.
During a delay $d$, both valuations advance by $d$, during an automaton discrete transition, only the automaton clocks may be reset, while formula clocks are unchanged. Accordingly, $\Paths_{\tts}(\statein)^+$ denotes the terminating finite runs of this extended semantics starting at $\statein$.

\begin{definition}[\ETOL semantics]\label{def:sat}
   % Let $\Clocks$ be the set of clocks of $\TA$ and $\ClocksFormula$ a non-empty set of clocks of the formula, $\atomic \in \Atomic, \ \guard \in \Constraints(\Clocks \cup \ClocksFormula)$.  
   The satisfaction relation between $\tts$, $\set{\prop, \propOther}\subseteq\ETOL$ and an extended state $\statein  = (\loc, \clockval, \clockvalformula)$, denoted $\tts, \statein \models \prop$, is defined inductively as follows:

    \begin{enumerate}
        \item $\tts, \statein\models \top$ for all state $s$,
        \item  $\tts, \statein\models \atomic$ \iff $\atomic\in \Label_{\tts}(\statein)$,%\dm{$\Label$ non défini ici, on est en TTS, voir ce qu'il faut mettre}
        \item  $\tts,\statein \models \neg \prop$ \iff not $\tts, s\models \prop$ (notation $\tts, \statein \not \models \prop$), 
        \item  $\tts, \statein \models \prop_1 \wedge \prop_2$ \iff $\tts, s\models \prop_1$ and $\tts, s\models \prop_2$, 
        \item  $\tts, \statein \models \guard$ \iff  $(\clockval \cup \clockvalformula) \models \guard$,
        \item  $\tts, \statein \models j.\prop$ \iff  $\tts, (\loc, \clockval,  \reset{\clockvalformula}{\{\clockformula\}}) \models \prop$,
        \item $\tts,\statein \models \ETexists (\prop_1 \until \prop_2)$ \iff there exist runs $\run , \run' \in \Paths_{\tts}(\statein)^{+}$ such that $\dur(\run)=\dur(\run') \ \land \ \run \models (\prop_1 \until \prop_2) \ \land \ \run' \not\models (\prop_1 \until \prop_2)$,
        %$(\exists i\geq 0$ s.t. $\tts,\run_i \models \prop_2$ and $\forall j<i, \ \tts,\run_j \models \prop_1)$ and $\exists \run' \in \Paths_{\tts}(\state)^{+}$ s.t. $\dur(\run')=\dur(\run)$ and $(\forall i' \geq 0,\ \tts,\run'_{i'} \not\models \prop_2$ or $\ \exists i' \geq 0 \ s.t. \ \tts,\run'_{i'} \models \prop_2 $ and $\exists j'<i'.\ \tts,\run'_{j'} \not\models \prop_1)$,
         \item $\tts,\statein \models \ETexists (\prop_1 \release \prop_2)$ \iff there exist runs $\run, \run' \in \Paths_{\tts}(\statein)^{+}$ such that $\dur(\run)=\dur(\run') \ \land \ \run \models (\prop_1 \release \prop_2) \ \land \ \run' \not\models (\prop_1 \release \prop_2)$,
         %$((\forall i\geq 0,\ \tts,\run_i \models \prop_2)$ or   $(\exists i\geq 0$ s.t. $\tts,\run_i \models \prop_1$ and  $\tts,\run_i \models \prop_2$ and $\forall j<i,\ \tts,\run_j \models \prop_2)$ and $\exists \run' \in \Paths_{\tts}(\state)^{+}$  s.t. $\dur(\run')=\dur(\run)$ and  $\exists i'\geq 0$  s.t. $(\tts,\run'_{i'} \not\models \prop_2$ and $\forall j'<i',\ \tts,\run'_{j'} \not\models \prop_1)$,
        \item $\tts,\statein \models \ETforall (\prop_1 \until \prop_2)$ \iff for every run $\run \in \Paths_{\tts}(\statein)^+$, there exists a run $\run' \in \Paths_{\tts}(\statein)^+$ such that $\dur(\run)=\dur(\run') \ \land \ \bigl(\run \models (\prop_1 \until \prop_2) \Leftrightarrow \run' \not\models (\prop_1 \until \prop_2) \bigr)$, 
        %$\forall \run \in \Paths_{\tts}(\state)^{+}$ s.t. $(((\exists i\geq 0$ s.t. $\tts,\run_i \models \prop_2$ and $\forall j<i, \tts,\run_j \models \prop_1\big)$ then $\exists \run' \in \Paths_{\tts}(\state)^{+}$ s.t. $\dur(\run')=\dur(\run)$  and $ (\forall i'\geq 0,\ \tts,\run'_{i'} \not\models \prop_2$ or  $\exists i'\geq 0$ s.t. $\tts,\run'_{i'} \models \prop_2$ and $\exists j'<i'$ s.t. $\tts,\run'_{j'} \not\models \prop_1))$ and $(\neg(\prop_1 \until \prop_2)$ \text{along the path} $\run$ then $ \exists \run' \in \Paths_{\tts}(\state)^{+}$ s.t. $ \dur(\run')=\dur(\run)$ and $(\prop_1 \until \prop_2)$ \text{along the path} $\run'))$,
         \item $\tts,\statein \models \ETforall (\prop_1 \release \prop_2)$ \iff for every run $\run \in \Paths_{\tts}(\statein)^+$, there exists a run $\run' \in \Paths_{\tts}(\statein)^+$ such that $\dur(\run)=\dur(\run') \ \land \ \bigl(\run \models (\prop_1 \release \prop_2) \Leftrightarrow \run' \not\models (\prop_1 \release \prop_2) \bigr)$
         
         %$\forall \run \in \Paths_{\tts}(\state)^{+}$ s.t. $(((\forall i\geq 0,\ \tts,\run_i \models \prop_2$ or $\exists i\geq 0$ s.t. $\tts,\run_i \models \prop_1$ and $\tts,\run_i \models \prop_2$ and $\forall j<i,\ \tts,\run_j \models \prop_2\Big)$ then $\exists \run' \in \Paths_{\tts}(\state)^{+}$ s.t. $ \dur(\run')=\dur(\run)$ and $\exists i'\geq 0$ s.t. $\tts,\run'_{i'} \not\models \prop_2$ and $\forall j'<i',\ \tts,\run'_{j'} \not\models \prop_1)$ and $ ((\exists i\geq 0$ s.t. $\tts,\run_i \not\models \prop_2$ and $\forall j<i,\ \tts,\run_j \not\models \prop_1)$ and $\exists \run' \in \Paths_{\tts}(\state)^{+}$ s.t. $\dur(\run')=\dur(\run)$ and $( \forall i'\geq 0,\ \tts,\run'_{i'} \models \prop_2$ or $\exists i'\geq 0$ s.t.  $\tts,\run'_{i'} \models \prop_1$ and $\tts,\run'_{i'} \models \prop_2$ and $\forall j'<i',\ \tts,\run'_{j'} \models \prop_2)))$.
    \end{enumerate}
\end{definition}

\noindent A formula $\prop$ is true in $\tts$ iff $\tts,\initialstate\models\prop$. Two formulas $\prop$ and $\propOther$ are equivalent, written $\prop\equiv\propOther$, iff for every  $\tts$ and state $\statein$ of $\tts$, $\tts,\statein\models\prop$ iff $\tts,\statein\models\propOther$. For any formula $\prop$, $\sat(\prop)$ denotes the set of states of $\tts$ satisfying $\prop$, namely $\sat(\prop)=\Set{\statein\in\states}{\tts,\statein\models\prop}$. The relationship between \iac{ta} and its semantics is defined as follows:

%\noindent A formula $\prop$ is true in $\tts$ \iff{} $\tts,\initialstate \models \prop$. Two formulas $\prop$ and $\propOther$ are equivalent (denoted by $\prop \equiv \propOther$) \iff{} for any model $\calM$ and state $\state$ of $\calM$,  $\calM, \state \models \prop$ iff $\calM, \state \models \propOther$. Let $\prop$ be any formula, then $\sat(\prop)$ denotes the set of states of $\tts$ verifying $\prop$, \ie{} $\sat(\prop)=\Set{\state \in \states}{\tts, \state \models \prop}$. The relationship between \iac{ta} and its semantics is defined as follows:

\begin{definition}
	Let $\TA$ be \iac{TA} and $\prop\in \ETOL$, then $\TA \models \prop$ \iff{} $\tts \models \prop$.
\end{definition}

\section{Model Checking}\label{sec:mc}
%%%%%%%%%%%%%%%%%%%%%%%%%%%%%%%%%%%%%%%%%%%%%%%%%%%%%%%%%%%%%%%%%%%%%%
%%%%%%%%%%%%%%%%%%%%%%%%%%%%%%%%%%%%%%%%%%%%%%%%%%%%%%%%%%%%%%%%%%%%%%
%Here, we present our model checking algorithm for \ETOL{} and show that the model checking problem for \ETOL{} is \PSPACE. Our algorithm follows a zone-based approach, extending techniques used for \ac{TCTL} and opacity verification in timed systems. The goal is to determine whether a given \ac{tts} $\tts$ satisfies and \ETOL{} formula. In this section, we fix an \ETOL formula $\prop$.
%In this section, we present a model-checking algorithm for \ETOL{} and prove that its model-checking problem is \PSPACE. The algorithm is based on a symbolic exploration of the timed state space through the zone-graph abstraction, combined with fixed-point computations for the branching-time temporal operators of \ETOL{}. Its purpose is to decide, for a given \ac{tts} $\tts$ and an \ETOL{} formula $\prop$, whether $\tts \models \prop$. For the remainder of this section, we fix an \ETOL{} formula $\prop$.
This section presents a symbolic model checking algorithm for \ETOL{} and proves that the problem is in \PSPACE. The algorithm explores the timed state space via the zone graph and uses fixed-point computations for the branching-time operators of \ETOL{}. Given a \ac{tts} $\tts$ and a formula $\prop$, it decides whether $\tts\models\prop$. In the sequel, we fix an \ETOL{} formula $\prop$.
%\vspace{-0.2cm}
\subsection{Observer clock and execution duration}
Before presenting the model checking algorithm, we extend the \ac{ta} with a distinguished observer clock $\clockMC$, used only to record the total elapsed duration. Formally, $\TA\oplus\clockMC = (\Loc,\locinit,\Clocks\cup\set{\clockMC},\Alphabet,\Invariant,\Label,\Final,\Transitions)$. The clock $\clockMC$ is initialized to $0$, evolves at the same
rate as the automaton clocks, and is never reset. Moreover, $\clockMC$ does not occur in guards or invariants of $\TA$. Consequently, it does not affect the enabledness of transitions or the behavior of the original automaton. Every finite run $\rho$ of $\TA$ has a unique extension to $\TA\oplus\clockMC$, and at the endpoint of this extension, $\clockval(\clockMC)=\dur(\rho)$. Conversely, projecting away $\clockMC$ recovers the corresponding run of $\TA$. Hence, two executions have the same observable execution time exactly when their endpoint valuations agree on
$\clockMC$. Equality is required only at the endpoints; their intermediate transitions need not occur at the same absolute times.

\begin{algorithm}[th!]
    \caption{\ETOL Model Checking }
    \label{alg:labeling}    
   % \dminline{Fix avec TCTL quand ce sera étendu}
  %  \dminline{Vérifier le modèle, on avait $\zonegraph(\TA, \prop)$, mais on parle jamais d'un ZG qui dépend d'une prop...}
\begin{algorithmic}[1] 
    \Require A \ac{TA} \TA, \ an \ETOL formula $\prop$, and a horizon $H$.
    \Ensure For any formula $\propOther$ of $\prop$: $\sat(\propOther)$ $\gets$$\Set{\statein\in \SymStates}{\zonegraph_H,\statein\models\propOther}$. \\
    
    \textbf{Compute} $\zonegraph_H(\TA \oplus \clockMC,\prop)$
    \ForAll{ $i$ $\le |\prop|$ } 
        \ForAll{$\propOther \in \sub(\prop) \text{ with } |\propOther| = i$} \Switch{$(\propOther)$} 
    
        \Case{$\propOther=\top$}
            \State{$ \sat(\propOther) \gets \SymStates$} 
        \EndCase 

        \Case{$\propOther=p$}
            \State{$ \sat(\propOther) \gets \Set{\statein \in \SymStates}{ p \in \Label_\ttz(s)}$} 
        \EndCase 

        \Case{$\propOther=\neg \propOther_1$} 
            \State{$\sat(\propOther) \gets \SymStates \setminus \sat(\propOther_1) $} 
        \EndCase 

        \Case{$\propOther= \guard$} 
            \State{$\sat(\propOther) \gets \Set{\statein \in \SymStates}{\forall \clockval \in \zone,\ \clockval \models \guard}$}
        \EndCase 
    
        \Case{$\propOther=\propOther_1 \land \propOther_2$}
            \State{$\sat(\propOther) \gets \sat(\propOther_1) \cap \sat(\propOther_2)$} 
        \EndCase 
  %  \Case{$\propOther=\OpOperator(\propOther_1 \until \propOther_2)$} \State{$\sat(\propOther) \gets \abbrv{O}\abbrv{U}(\sat(\propOther_1), \sat(\propOther_2))$} 
  %  \EndCase 
  %  \Case{$\propOther=\OpOperator(\propOther_1 \release \propOther_2)$} \State{$\sat(\propOther) \gets \abbrv{O}\abbrv{R}(\sat(\propOther_1), \sat(\propOther_2))$} 
  %  \EndCase
 
    \Case{$\propOther=j.\propOther_1$}
      \State{$\sat(\propOther) \gets \Set{\statein \in \SymStates}{\reset{\clockvalformula}{\{\clockformula\}}(\statein) \in \sat(\propOther_1)}$}
    \EndCase 
    
    \Case{$\propOther=\ETexists (\propOther_1 \until \propOther_2)$}
       % \State $\sat(\propOther) \gets \procedureOUexists \big(\sat(\propOther_1),\sat(\propOther_2)\big)$
         \State $U \gets \procedureOUexists \big(\sat(\propOther_1),\sat(\propOther_2)\big)$
    \State $\sat(\propOther) \gets \matchOU \big(U,\SymStates\setminus U\big)$
    \EndCase

    \Case{$\propOther=\ETforall (\propOther_1 \until \propOther_2)$}
       % \State $\sat(\propOther) \gets \procedureOUforall \ \big(\sat(\propOther_1),\sat(\propOther_2)\big)$
        \State $U \gets \procedureOUforall \big(\sat(\propOther_1),\sat(\propOther_2)\big)$
    \State $\sat(\propOther) \gets \matchOU \big(U,\SymStates\setminus U\big)$
    \EndCase

    \Case{$\propOther=\ETexists (\propOther_1 \release \propOther_2)$}
       % \State $\sat(\propOther) \gets \procedureORexists \ \big(\sat(\propOther_1),\sat(\propOther_2)\big)$
        \State $R \gets \procedureORexists  \big(\sat(\propOther_1),\sat(\propOther_2)\big)$
    \State $\sat(\propOther) \gets  \matchOR \big(R,\SymStates\setminus R\big)$
    \EndCase

    \Case{$\propOther=\ETforall (\propOther_1 \release \propOther_2)$}
       % \State $\sat(\propOther) \gets \procedureORforall \ \big(\sat(\propOther_1),\sat(\propOther_2)\big)$
        \State $R \gets \procedureORforall \big(\sat(\propOther_1),\sat(\propOther_2)\big)$
    \State $\sat(\propOther) \gets \matchOR \big(R,\SymStates\setminus R\big)$
    \EndCase
    \EndSwitch 
    \EndFor 
    \EndFor 
    \end{algorithmic}  
\end{algorithm}

\subsection{Bounded symbolic abstraction}
\label{subsec:bounded-abstraction}

The observer clock $\clockMC$ records the exact elapsed execution time and is therefore neither extrapolated  nor saturated~\cite{BOY04}, since merging values above a constant could identify distinct durations. Zones are represented by DBMs over $\Clocks\cup\{\clockMC\}$, with standard extrapolation applied only to the ordinary automaton clocks. For bounded model checking, we fix a verification horizon $H\in\Rplus$ and restrict the exploration to runs satisfying $\clockMC\leq H$, equivalently $\dur(\rho)\leq H$. Thus, $\clockMC$ remains exact while its range is bounded. Combined with the finite abstraction of the ordinary clocks, this yields a finite symbolic exploration for each fixed $H$. Bounding and saturation must be distinguished: the former discards executions longer than $H$, whereas the latter merges distinct execution times. Our construction uses only the former. For convenience, we write $\Paths_{\tts}^{\leq H}(s) = \{\rho\in\Paths_{\tts}(s)^+ \mid \dur(\rho)\leq H\}$.

\subsection{Algorithm overview}
The algorithm works over the finite bounded zone graph
$\zonegraph_H(\TA\oplus\clockMC)$. Ordinary automaton clocks are normalized using the standard finite zone abstraction, whereas the observer clock $\clockMC$ is kept exact and restricted to $[0,H]$. Hence, finiteness follows from the combination of the standard abstraction on $\Clocks$ and the finite verification horizon  on $\clockMC$.
\subsection{Algorithm details}
\cref{alg:labeling} takes as input a \ac{ta} $\TA$, an \ETOL{} formula $\prop$, and the bounded zone graph $\zonegraph_H(\TA\oplus\clockMC)$. It evaluates $\prop$ bottom-up by computing, for every subformula
$\propOther$, $\sat(\propOther)
= \Set{\SymState\in\SymStates}
{\zonegraph_H,\SymState\models\propOther}$. Atomic propositions, clock constraints, and Boolean operators are handled by standard symbolic set operations. The opacity-temporal cases require additional processing. The procedures $\procedureOUexists$ and $\procedureOUforall$ compute the temporal satisfaction sets for until formulas, while $\procedureORexists$ and
$\procedureORforall$ handle release formulas. These sets are then passed to the corresponding execution-time matching procedures, which compare satisfying and non-satisfying executions having the same total duration.
The procedures are detailed in \cref{sec:opacityprocedures}.
%\vspace{-0.4cm}

% \subsection{Particular procedures for opacity operators}
%\vspace{-1.3cm}
\subsection{Execution-time predecessor operators}

%To evaluate execution-time opacity modalities, we introduce two symbolic predecessor operators over the zone graph. Given $\setX\subseteq\SymStates$, the operator $\Pre_{\ETexists}(\setX)$ collects predecessor states satisfying the existential opacity condition, while $\Pre_{\ETforall}(\setX)$ collects those satisfying the universal one. Intuitively, a symbolic state belongs to such a predecessor set when it admits two alternative continuations: one towards $\setX$ and one towards $\SymStates\setminus\setX$, with indistinguishable observable execution times. For $\Pre_{\ETexists}$, it is enough that some common duration is possible for both alternatives. For $\Pre_{\ETforall}$, every duration possible on one side must be matched by a duration on the other side, and conversely. These duration comparisons are performed symbolically using the global clock $\clockMC$, which records the total elapsed execution time. For $\setX\subseteq\SymStates$ and $\symZone\in\SymStates$, we write $\Post_{\setX}(\symZone)=\Post(\{\symZone\})\cap\setX$. 

To evaluate ET-opacity modalities, we introduce two symbolic predecessor operators over the bounded zone graph. Given $\setX\subseteq\SymStates$, the operator $\Pre_{\ETexists}(\setX)$ collects symbolic states admitting two alternative one-step
continuations, one towards $\setX$ and one towards $\SymStates\setminus\setX$, that can be realized with the same elapsed execution time. The universal operator $\Pre_{\ETforall}(\setX)$
requires the corresponding duration-matching condition in both directions. Duration equality is checked using the observer clock $\clockMC$,
which records the exact elapsed execution time and is neither extrapolated nor saturated. %$\Post_{\setX}(\symZone) = \Post(\{\symZone\})\cap\setX$ and 
For $\setX\subseteq\SymStates$ and
$\symZone\in\SymStates$, we write  $\Post_{\setX}^{+}(\symZone) = \left\{ \symZone'\in\setX \mid \symZone\Rightarrow^{+}\symZone' \right\}$ where $\Rightarrow^{*}$ denotes the reflexive-transitive closure of the symbolic transition relation. For a symbolic state $\symZone\in\SymStates$, we define its concretization as $\gamma(\symZone)
= \left\{ (\loc,\clockval) \mid \clockval\in\zone
\right\}$. For a set of symbolic states $\setX\subseteq\SymStates$, we extend $\gamma$ pointwise as $\gamma(\setX) = \bigcup_{\symZone\in\setX}
\gamma(\symZone)$.

\begin{lemma}
\label{lem:observer-equal-duration}
Let $\rho$ and $\rho'$ start from the same concrete state
$(\loc,\clockval)$, and let $\last(\rho)=(\loc',\clockval')$, $\last(\rho')=(\loc'',\clockval'')$. Since $\clockMC$ is never reset, $\clockval'(\clockMC)=\clockval(\clockMC)+\dur(\rho)$
and $\clockval''(\clockMC)=\clockval(\clockMC)+\dur(\rho')$. Consequently, $\clockval'(\clockMC)=\clockval''(\clockMC) \Longleftrightarrow \dur(\rho)=\dur(\rho')$.
\end{lemma}

\begin{proof}
The observer clock $\clockMC$ evolves at rate $1$ during every delay transition and is unchanged by every discrete transition. Moreover, $\clockMC$ is never reset. Hence, along any execution starting from $(\loc,\clockval)$, the value of $\clockMC$ increases exactly by the total amount of elapsed time. Therefore,
$\clockval'(\clockMC) = \clockval(\clockMC)+\dur(\rho)$
and similarly $\clockval''(\clockMC) = \clockval(\clockMC)+\dur(\rho')$. Since $\rho$ and $\rho'$ start from the same valuation $\clockval$, both expressions contain the same initial value $\clockval(\clockMC)$. Thus, $\clockval'(\clockMC)=\clockval''(\clockMC)$ holds iff
$\dur(\rho)=\dur(\rho')$.
\end{proof}

\begin{definition}
Let  $\setX \subseteq \SymStates$ be  a set of symbolic states and $\symZone= (\loc,\zone)$, we write:
$$\Pre_{\ETexists} (\setX) = \Set{\symZone \in \SymStates}  { \left.
 \begin{cases*}
     % \symZone \in (\Pre(\setX) \cap \Pre(\SymStates \setminus \setX)),  \   \\
        \exists \symZone' \in \Post_{\setX}^{+}(\symZone), \ \exists \symZone'' \in \Post_{\SymStates\setminus\setX}^{+}(\symZone), \\
        \exists \clockval \in \zone,\
\exists \clockval' \in \zone',\
\exists \clockval'' \in \zone'' \\
\exists\,\rho,\rho'
\text{ starting from }(\loc,\clockval)
\text{ such that} \\
\last(\rho)=(\loc',\clockval'),
\last(\rho')=(\loc'',\clockval''),
\\
\clockval'(\clockMC)
=
\clockval''(\clockMC)
\leq H 
\end{cases*}
   \right\} }
$$
\end{definition}

%\begin{proposition} Let $\setX\subseteq\SymStates$ be a set of symbolic states. We define $\Pre_{\ETexists}(\setX)$ as the set symbolic states $\symZone=(\loc,\zone)\in\SymStates$ such that $\symZone \in \Pre(\setX)\cap \Pre(\SymStates\setminus\setX)$,   $\symZone'\in\Post_{\setX}^{+}(\symZone)$, $\symZone''\in \Post_{\SymStates\setminus\setX}^{+}(\symZone)$, 
%$\exists \clockval\in\zone,
%\exists \clockval'\in\zone',
%\exists \clockval''\in\zone''$, 
%$\exists\,\rho,\rho'
%\text{ starting from }(\loc,\clockval)
%\text{ such that}
%\last(\rho)=(\loc',\clockval'),
%\last(\rho')=(\loc'',\clockval''), \clockval'(\clockMC)= \clockval''(\clockMC) \leq H $.
%\end{proposition}

\begin{proposition}[Correctness of $\Pre_{\ETexists}$]
\label{prop:pre-etexists-correctness}
Let $\setX\subseteq\SymStates$ and let $\symZone=(\loc,\zone)$.
Then $\symZone\in\Pre_{\ETexists}(\setX)$ iff there exists a concrete valuation $\clockval\in\zone$ and two
non-empty executions $\rho$ and $\rho'$ starting from
$(\loc,\clockval)$ such that $\last(\rho)\in\gamma(\setX), \last(\rho')\in\gamma(\SymStates\setminus\setX)$,
and $\dur(\rho)=\dur(\rho')\leq H$.
\end{proposition}

\begin{proof}
Assume first that $\symZone\in\Pre_{\ETexists}(\setX)$. By definition, there exist $\symZone'=(\loc',\zone')
\in\Post_{\setX}^{+}(\symZone)$ and
$\symZone''=(\loc'',\zone'') \in\Post_{\SymStates\setminus\setX}^{+}(\symZone)$,
together with valuations $\clockval\in\zone, \clockval'\in\zone', \clockval''\in\zone''$, and concrete executions $\rho,\rho'$ starting from
$(\loc,\clockval)$ such that
$\last(\rho)=(\loc',\clockval')$, $\last(\rho')=(\loc'',\clockval'')$, and
$\clockval'(\clockMC) = \clockval''(\clockMC)
\leq H$. By \cref{lem:observer-equal-duration},
$\clockval'(\clockMC) = \clockval''(\clockMC)$ is equivalent to $\dur(\rho)=\dur(\rho')$. Thus the required executions exist. Conversely, assume that there are two such executions $\rho$ and $\rho'$
starting from the same concrete state $(\loc,\clockval)$, with $\dur(\rho)=\dur(\rho')\leq H$, such that their endpoints are represented respectively by a symbolic state in $\setX$ and one in $\SymStates\setminus\setX$. By completeness of the bounded symbolic construction, both executions are
represented by symbolic paths from $\symZone$. Therefore their endpoint symbolic states belong respectively to
$\Post_{\setX}^{+}(\symZone)$
and $\Post_{\SymStates\setminus\setX}^{+}(\symZone)$.
Finally, by \cref{lem:observer-equal-duration}, the equality of durations implies $\clockval'(\clockMC)
= \clockval''(\clockMC)$. Hence all conditions in the definition of $\Pre_{\ETexists}$ are satisfied, and therefore $\symZone\in\Pre_{\ETexists}(\setX)$.
\end{proof}

\begin{definition}
Let $\setX$ $\subseteq$ $\SymStates$  be  a set of symbolic states and $\symZone= (\loc,\zone)$, we write:
$$
\Pre_{\ETforall} (\setX)
=
\Set{\symZone \in \SymStates}{ \left.
\begin{cases*}
%\symZone \in (\Pre(\setX) \cap \Pre(\SymStates \setminus \setX)),  \\
\forall \symZone' \in \Post_{\setX}^{+}(\symZone),   \exists \symZone'' \in \Post_{\SymStates\setminus\setX}^{+}(\symZone) \text{ s.t } \\
\exists \clockval\in\zone,\ \exists \clockval'\in\zone',\  \exists \clockval''\in\zone''  \\
\exists\,\rho,\rho'
\text{ starting from }(\loc,\clockval)
\text{ such that} \\
\last(\rho)=(\loc',\clockval'),
\last(\rho')=(\loc'',\clockval''),
\\
\clockval'(\clockMC)
=
\clockval''(\clockMC)
\leq H  \\
\ \land \\
\forall \symZone'' \in \Post_{\SymStates\setminus\setX}^{+}(\symZone),  \exists \symZone'\in \Post_{\setX}^{+}(\symZone) \text{ s.t} \\ 
\exists \clockval\in\zone,\ \exists \clockval'\in\zone',\ \exists \clockval''\in\zone''  \\
\exists\,\rho,\rho'
\text{ starting from }(\loc,\clockval)
\text{ such that} \\
\last(\rho)=(\loc',\clockval'),
\last(\rho')=(\loc'',\clockval''),
\\
\clockval'(\clockMC)
=
\clockval''(\clockMC)
\leq H 
\end{cases*}
\right\} }
$$

%and for a set $\setX$, we define $\PreForall(\setX) = \Set{ \PreForall(\symZone)}{\symZone \in \setX}$.
\end{definition}

\begin{proposition}[Correctness of $\Pre_{\ETforall}$]
\label{prop:pre-etforall-correctness}
Let $\setX\subseteq\SymStates$ and let $\symZone=(\loc,\zone)$. Then $\symZone\in\Pre_{\ETforall}(\setX)$ iff, for every $\clockval\in\zone$, the following two conditions hold: 
\begin{enumerate}
    \item every execution from $(\loc,\clockval)$ ending in $\gamma(\setX)$  has an execution of the same duration ending in $\gamma(\SymStates\setminus\setX)$,
    \item every execution from $(\loc,\clockval)$ ending in  $\gamma(\SymStates\setminus\setX)$ has an execution of the same duration ending in $\gamma(\setX)$.
\end{enumerate}

\end{proposition}

\begin{proof}
Assume $\symZone\in\Pre_{\ETforall}(\setX)$. Let $\clockval\in\zone$ and let $\rho$ be any execution starting from $(\loc,\clockval)$ whose endpoint belongs to $\gamma(\setX)$. By the definition of $\Pre_{\ETforall}$, there exists an execution $\rho'$ from the same concrete state whose endpoint belongs to $\gamma(\SymStates\setminus\setX)$ and whose final observer-clock value coincides with that of $\rho$:
$\clockval'(\clockMC)= \clockval''(\clockMC)$. By \cref{lem:observer-equal-duration}, $\dur(\rho)=\dur(\rho')$.
This proves (i). The symmetric clause in the definition of
$\Pre_{\ETforall}$ gives (ii) in exactly the same way. Conversely, suppose that conditions (i) and (ii) hold for every
$\clockval\in\zone$. Consider any execution represented by a symbolic path from $\symZone$ whose endpoint belongs to $\setX$. By (i), there exists an execution from the same concrete initial state ending outside $\setX$ with the same total duration. By \cref{lem:observer-equal-duration}, their final observer-clock values are equal. The same reasoning using (ii) applies to executions ending outside $\setX$. Thus both matching requirements in the definition of
$\Pre_{\ETforall}$ hold, and hence $\symZone\in\Pre_{\ETforall}(\setX)$.
\end{proof}

\begin{proposition}[Finiteness]
\label{prop:finiteness}
For every \ac{TA} $\TA$, \ETOL formula $\prop$, and fixed verification horizon $H$, the bounded symbolic graph
$\zonegraph_H(\TA\oplus\clockMC,\prop)$ constructed using an extrapolation abstraction ($Extra_M^H$) is finite.
\end{proposition}

\begin{proof}
The automaton $\TA$ contains finitely many locations and finitely many clocks. Under $Extra_M^H$, ordinary automaton clocks are extrapolated with respect to the finite set of maximal constants, whereas the observer clock $\clockMC$ is kept exact and restricted to $0\leq\clockMC\leq H$. By the finite-bound lemma, only finitely many canonical DBMs can occur in the bounded symbolic construction. Since every symbolic state consists of a location together with one of these canonical DBMs, only finitely many symbolic states are reachable. Therefore, $\zonegraph_H(\TA\oplus\clockMC,\prop)$
is finite.
\end{proof}

\subsection{Procedures}
\label{sec:opacityprocedures}

%As in standard zone-based model checking for branching-time timed logics, \ac{etopacity} operators are evaluated by fixed-point computations over symbolic predecessor sets. For until formulas, the iteration starts from $\sat(\propOther_2)$ and repeatedly adds states that satisfy $\propOther_1$ and can reach the current approximation through the appropriate execution-time predecessor operator. The existential and universal cases are handled by $\procedureOUexists$ and $\procedureOUforall$, respectively (\cref{alg:ou_exist,alg:ou_forall}). Release formulas are treated dually, using greatest fixed-point computations based on the same predecessor operators. The procedures $\procedureORexists$ and $\procedureORforall$ are given in \cref{alg:or_exist,alg:or_forall}.
Opacity-temporal formulas are evaluated in two phases. First, we compute the
ordinary temporal satisfaction set of the underlying until or release formula.
The until cases are handled by $\procedureOUexists$ and $\procedureOUforall$
(\cref{alg:ou_exist,alg:ou_forall}), while the release cases are handled by
$\procedureORexists$ and $\procedureORforall$
(\cref{alg:or_exist,alg:or_forall}). Second, we apply the corresponding
execution-time matching test between the computed set and its complement: $\matchOU$ for until formulas, and
$\matchOR$ for release formulas.

\begin{algorithm}[th!]
    \caption{Backward search for computing  $ \procedureOUexists$ }
    \label{alg:ou_exist} 
\begin{algorithmic}[1] 
    \Require An \ETOL formula $\ETexists(\propOther_1 \until \propOther_2)$.
    \Ensure  $\sat(\propOther_1 \until \propOther_2) \gets \Set{\statein \in \states}{\zonegraph_H, \statein \models \ETexists (\propOther_1 \until \propOther_2) }$.

    \State{$ X \gets \emptyset$}
        \State{$ Y \gets \sat(\propOther_2)$}

        \While{$Y \neq X$}
           \State{$X \gets Y$}
           \State{$Y \gets \sat(\propOther_2) \cup (\sat(\propOther_1) \cap \PreExist(X))$ }
        \EndWhile
    \State{$ \mathbf{return} \ Y$}
\end{algorithmic}
\end{algorithm}
\begin{algorithm}[th!]
    \caption{Backward search for computing  $ \procedureOUexists$ }
    \label{alg:ou_exist} 
\begin{algorithmic}[1] 
    \Require An \ETOL formula $\ETexists(\propOther_1 \until \propOther_2)$.
    \Ensure  $\sat(\propOther_1 \until \propOther_2) \gets \Set{\statein \in \states}{\zonegraph_H, \statein \models \ETexists (\propOther_1 \until \propOther_2) }$.

    \State{$ X \gets \emptyset$}
        \State{$ Y \gets \sat(\propOther_2)$}

        \While{$Y \neq X$}
           \State{$X \gets Y$}
           \State{$Y \gets \sat(\propOther_2) \cup (\sat(\propOther_1) \cap \PreExist(X))$ }
        \EndWhile
    \State{$ \mathbf{return} \ Y$}
\end{algorithmic}
\end{algorithm}

\begin{algorithm}[th!]
\caption{Backward search for computing  $ \procedureORexists$} \label{alg:or_exist} 
    \begin{algorithmic}[1] 
    \Require An \ETOL formula $\ETexists(\propOther_1 \release \propOther_2)$.
    \Ensure $\sat(\propOther_1 \release \propOther_2) \gets \{s \in S \mid \zonegraph_H, s \models \ETexists \ (\propOther_1 \release \propOther_2) \}$.
    \State{$ X \gets \emptyset$} \State{$ Y \gets \sat(\propOther_2)$} 
    \While{$Y \neq X$} \State{$X \gets Y$} \State{$Y \gets \propOther_2 \cap (\propOther_1 \cup \Pre_{\ETexists} (X))$ } 
    \EndWhile \State{$ \mathbf{return} \ Y$} 
    \end{algorithmic} 
\end{algorithm}
%----------------------------------------------------------------
%----------------------------------------------------------------
\begin{algorithm}[th!]
\caption{Backward search for computing $\procedureOUexists$}
\label{alg:ou_forall}
\begin{algorithmic}[1]
\Require An \ETOL subformula $\ETforall(\propOther_1 \until \propOther_2)$.
\Ensure $\sat\big(\propOther_1 \until \propOther_2\big) \gets \Set{\statein \in \states}{\zonegraph_H, \statein \models \ETforall (\propOther_1 \until \propOther_2)}$.
\State $X \gets \emptyset$
\State $Y \gets \sat(\propOther_2)$
\While{$Y \neq X$}
  \State $X \gets Y$
  \State $Y \gets \sat(\propOther_2)\ \cup\ \big(\sat(\propOther_1)\ \cap\ \Pre_{\ETforall} \ (X)\big)$
\EndWhile
\State \Return $Y$
\end{algorithmic}
\end{algorithm}
%----------------------------------------------------------------
%----------------------------------------------------------------
\begin{algorithm}[th!]
\caption{Backward search for computing $\procedureORforall$}
\label{alg:or_forall}
\begin{algorithmic}[1]
\Require An \ETOL subformula $\ETforall(\propOther_1 \release \propOther_2)$.
\Ensure $\sat\big(\propOther_1 \release \propOther_2\big) \gets \Set{\statein \in \states}{\zonegraph_H, \statein \models \ETforall (\propOther_1 \release \propOther_2) } $.
\State $X \gets \emptyset$
\State $Y \gets \sat(\propOther_2)$
\While{$Y \neq X$}
  \State $X \gets Y$
  \State $Y \gets \sat(\propOther_2)\ \cap\ \big(\sat(\propOther_1)\ \cup\ \PreForall(X)\big)$
\EndWhile
\State \Return $Y$
\end{algorithmic}
\end{algorithm}

\begin{algorithm}[th!]
\caption{Existential execution-time matching}
\label{alg:match-exists}
\begin{algorithmic}[1]
\Require A set $S\subseteq\SymStates$
\Ensure $\matchOU(S,\SymStates\setminus S)$
\State $M\gets\emptyset$
\ForAll{$\symZone=(\loc,\zone)\in\SymStates$}
    \If{there exist
    $\symZone'=(\loc',\zone')\in\Post^{+}_{S}(\symZone)$
    and
    $\symZone''=(\loc'',\zone'')\in
    \Post^{+}_{\SymStates\setminus S}(\symZone)$,
    together with
    $\clockval\in\zone$,
    $\clockval'\in\zone'$ and
    $\clockval''\in\zone''$ such that}
        \If{there exist concrete runs $\rho,\rho'$ starting from
        $(\loc,\clockval)$ with
        $\last(\rho)=(\loc',\clockval')$,
        $\last(\rho')=(\loc'',\clockval'')$, and
        $\clockval'(\clockMC)=\clockval''(\clockMC)\leq H$}
            \State $M\gets M\cup\{\symZone\}$
        \EndIf
    \EndIf
\EndFor
\State \Return $M$
\end{algorithmic}
\end{algorithm}

\begin{algorithm}[th!]
\caption{Universal execution-time matching}
\label{alg:match-forall}
\begin{algorithmic}[1]
\Require A set $S\subseteq\SymStates$
\Ensure $\matchOR(S,\SymStates\setminus S)$
\State $M\gets\emptyset$
\ForAll{$\symZone=(\loc,\zone)\in\SymStates$}
    \If{for every concrete run $\rho$ from a concrete state
    $(\loc,\clockval)$ represented by $\symZone$ and ending in $S$,
    there exists a concrete run $\rho'$ from the same
    $(\loc,\clockval)$ ending in $\SymStates\setminus S$ such that
    their final valuations satisfy
    $\clockval'(\clockMC)=\clockval''(\clockMC)\leq H$}
        \If{and conversely, every run ending in
        $\SymStates\setminus S$ has such a matching run ending in $S$}
            \State $M\gets M\cup\{\symZone\}$
        \EndIf
    \EndIf
\EndFor
\State \Return $M$
\end{algorithmic}
\end{algorithm}

%As in the standard zone-based model checking for branching-time timed logics, the \ac{etopacity} operators are evaluated by fixed-point computations over symbolic predecessor sets:
%\begin{itemize}
 %   \item For the $\until$ operator, the fixed-point iteration starts from $\sat(\propOther_2)$ and repeatedly adds symbolic states that satisfy $\propOther_1$ and can reach the current set via the appropriate predecessor operator. We present $\procedureOUexists$ in \cref{alg:ou_exist} and $\procedureOUforall$ in \cref{alg:ou_forall}.
%    \item The $\release$ operator is handled in a dual and analogous way, using a fixed-point computation that also relies on the corresponding predecessor operator. $\procedureORexists$ and $\procedureORforall$ are presented in \cref{alg:or_exist,alg:or_forall}, respectively.
%\end{itemize}
%\vspace{-0.4cm}

\begin{proposition}[Correctness of $\matchOU$]
\label{prop:matchU-correctness}
Let  $X=\sat(\phi)$ and $Y=\sat(\psi)$. Then
$\matchOU(X,Y) = \sat(\phi\until\psi)$.
\end{proposition}

\begin{proof}
The algorithm computes the least fixed point of
$F(Z) = Y\cup(X\cap\Pre(Z))$. We prove that this fixed point contains exactly the symbolic states from which there exists a finite symbolic path satisfying $\phi\until\psi$. Let
$Z_0=Y$ and, for $i\geq0$, $Z_{i+1} = Z_i\cup \bigl(X\cap\Pre(Z_i)\bigr)$. We show by induction on $i$ that $\symZone\in Z_i$ iff there exists a symbolic path of length at most $i$ from $\symZone$ to a state satisfying $\psi$, while every preceding state satisfies $\phi$. For $i=0$, we have $Z_0=Y$, so $\psi$ already holds at the current state. For the induction step, assume the claim for $Z_i$. A state is added to
$Z_{i+1}$ iff it belongs to $X$ and has a successor in $Z_i$. Hence $\phi$ holds at the current state and, by the induction hypothesis, the successor admits a path eventually reaching a $\psi$-state while maintaining $\phi$ beforehand. This is exactly the semantics of $\phi\until\psi$. Conversely, every finite path satisfying $\phi\until\psi$ has some finite
length $k$ before reaching a $\psi$-state. Repeated application of the predecessor rule therefore places its initial symbolic state in $Z_k$. Since the symbolic state space is finite, the increasing sequence
$Z_0\subseteq Z_1\subseteq\cdots$ stabilises after finitely many iterations. Its least fixed point is
therefore exactly $\sat(\phi\until\psi)$.
\end{proof}

\begin{proposition}[Correctness of $\matchOR$]
\label{prop:matchR-correctness}
Let $X=\sat(\phi)$ and $Y=\sat(\psi)$. Then
$\matchOR(X,Y) = \sat(\phi\release\psi)$.
\end{proposition}

\begin{proof}
The release operator is characterised by the greatest fixed point  $\nu Z.\;Y\cap\bigl(X\cup\Pre(Z)\bigr)$. Let
$F(Z) = Y\cap\bigl(X\cup\Pre(Z)\bigr)$. The algorithm $\matchOR$ computes the greatest fixed point of $F$. A symbolic state belongs to $F(Z)$ iff $\psi$ holds at the current state
and either $\phi$ already holds there or there exists a successor belonging to $Z$. Hence, as long as $\phi$ has not occurred, $\psi$ must continue to hold along the symbolic execution. This is precisely the semantics of
$\phi\release\psi$. Starting from the whole finite symbolic state space and repeatedly applying $F$ produces a decreasing sequence
$\SymStates \supseteq F(\SymStates) \supseteq F^2(\SymStates)
\supseteq\cdots$. Since $\SymStates$ is finite, this sequence stabilises after finitely many iterations at the greatest fixed point of $F$. Therefore, $\matchOR(X,Y) = \sat_H(\phi\release\psi)$.
\end{proof}

\subsection{Termination, soundness and complexity}

Termination of \cref{alg:labeling} follows from the finiteness of the \ac{zonegraph}. %The propositions below states both its termination and correctness.
%Termination of the \cref{alg:labeling} intuitively follows, as the number of states in the \ac{zonegraph} is finite. The following proposition establishes the termination and the correctness of our model checking algorithm.

\begin{proposition}[Termination] 
\label{prop:termination}
Let $\TA$ be \iac{ta} and let $\prop$ be an \ETOL formula. Let $H$ be a fixed verification horizon. Then \cref{alg:labeling} always terminates on input $\zonegraph_H(\TA\oplus\clockMC,\prop)$. 
\end{proposition} 

\begin{proof}
By \cref{prop:finiteness}, $\zonegraph_H(\TA\oplus\clockMC,\prop)$ contains finitely many symbolic states. Moreover, $\prop$ contains finitely many subformulas. Algorithm~\ref{alg:labeling} processes each subformula over the finite symbolic state space. The Boolean cases require only finite set operations. The procedures \matchOU \ and \matchOR \ compute monotone fixed points over the finite lattice $2^{\SymStates}$ and therefore stabilise after finitely many iterations. The operators
$\Pre_{\ETexists}$ and $\Pre_{\ETforall}$ range over the finite bounded symbolic graph and compare only executions represented in this graph. Hence every invocation made by Algorithm~\ref{alg:labeling} terminates. Therefore the whole labelling algorithm terminates.
\end{proof}

\begin{theorem}[Soundness]
\label{thm:soundness}
Let $\symZone$ be a symbolic state of
$\zonegraph_H(\TA\oplus\clockMC,\prop)$. If Algorithm~\ref{alg:labeling} labels $\symZone$ with $\prop$, then every concrete state represented by the corresponding satisfying subzone satisfies $\prop$ under the bounded semantics.
\end{theorem}
\begin{proof}
 The proof proceeds by structural induction on $\prop$. The cases of atomic propositions and clock constraints follow directly from their concrete and symbolic semantics. The Boolean cases follow from the induction hypothesis and the corresponding set operations. For
$\prop=\phi\until\psi$, soundness follows from
\cref{prop:matchU-correctness}. For  
$\prop=\phi\release\psi$, soundness follows from
\cref{prop:matchR-correctness}. For
$\prop=\ETexists\theta$, let $X=\sat_H(\theta)$. The algorithm labels the symbolic state according to
$\Pre_{\ETexists}(X)$. By
\cref{prop:pre-etexists-correctness}, there are two concrete executions from the same concrete initial state, one satisfying $\theta$ and the other satisfying its complement, with equal total duration. Hence the
concrete state satisfies $\ETexists\theta$. For
$\prop=\ETforall\theta$, the result follows analogously from \cref{prop:pre-etforall-correctness}: every execution satisfying one side admits a matching execution satisfying the other side with the same total
duration, and conversely. Thus every symbolic label produced by the algorithm is semantically
sound.   
\end{proof}
$
%\ETexists(\propOther_1 \until \propOther_2),
%\ETforall(\propOther_1 \until \propOther_2),
%\ETexists(\propOther_1 \release \propOther_2),
%\ETforall(\propOther_1 \release \propOther_2)$. By assumption, the procedures $\procedureOUexists$, $\procedureOUforall$,
%$\procedureORexists$, and $\procedureORforall$ compute exactly the corresponding symbolic satisfaction sets. Hence the algorithmic clauses coincide with the semantic clauses in all four cases. Therefore, for every subformula $\propOther \in \sub(\prop)$,
%$\sat(\propOther)
%=
%\Set{\state \in \SymStates}{\zonegraph,\state \models \propOther}.
%$
%Applying this to the full formula $\prop$ yields
$%\zonegraph,\initialSymState \models \prop$
%\iff
%$\initialSymState \in \sat(\prop),$
%which proves soundness and completeness.
%\end{proof}
%\begin{proposition}[Soundness and Completeness]
%    Let $\prop$ be \ETOL formula.
%    Then, \cref{alg:labeling} is sound and complete: for every extended initial state  $\state_0=(\locinit,\clockval_0,\clockvalformula_0)$, $\tts,\state_0 \models \prop \Longleftrightarrow \state_0 \in \sat(\prop)$
%\end{proposition}
%\begin{proof}
  %  \todo{}
%\end{proof}
\begin{theorem}[Completeness]
\label{thm:completeness}
Let $s=(\loc,\clockval)$ be a concrete state represented in $\zonegraph_H(\TA\oplus\clockMC,\prop)$. If
$s\models_H\prop$, then Algorithm~\ref{alg:labeling} represents $s$ in a symbolic subzone labeled with $\prop$.
\end{theorem}

\begin{proof}
The proof is by structural induction on $\prop$. The atomic and clock-constraint cases follow from the symbolic representation of valuations. The Boolean cases follow immediately from the induction hypothesis. If
$s\models_H\phi\until\psi$, there exists a finite execution witnessing the until condition. By completeness of the symbolic transition relation, this execution is represented by a symbolic path. By
\cref{prop:matchU-correctness}, the symbolic subzone containing $s$ is therefore returned by \matchOU. The release case follows analogously from
\cref{prop:matchR-correctness}. Suppose
$s\models_H\ETexists\theta$. By the semantics of $\ETexists$, there exist two executions $\rho$ and
$\rho'$ from $s$, one satisfying $\theta$ and the other its complement, such that $\dur(\rho)=\dur(\rho')\leq H$.
By completeness of the symbolic graph, both executions are represented by symbolic paths. By \cref{lem:observer-equal-duration}, their endpoint valuations satisfy
$\clockval'(\clockMC) = \clockval''(\clockMC)$. Hence the conditions of $\Pre_{\ETexists}$ are met, and the symbolic subzone containing $s$ is labelled with $\ETexists\theta$. The $\ETforall$ case is analogous. Every execution required by the universal semantics has a duration-equivalent counterpart. By symbolic
completeness and \cref{lem:observer-equal-duration}, all such matching pairs satisfy the condition used by $\Pre_{\ETforall}$. Hence the algorithm labels the symbolic representation of $s$ accordingly. Therefore the algorithm is complete.
\end{proof}

\begin{theorem}[Correctness of bounded \ETOL model checking]
\label{thm:correctness}
Let $\TA$ be a timed automaton, let $\prop$ be an \ETOL formula, and let $H$ be a fixed verification horizon.
Algorithm~\ref{alg:labeling} terminates and $\TA\models_H\prop
\quad\Longleftrightarrow\quad
s_0\models_H\prop$ is correctly decided by the labelling of the initial symbolic state.
\end{theorem}

\begin{proof}
Termination follows from \cref{prop:termination}.
Soundness follows from \cref{thm:soundness}, and completeness follows from \cref{thm:completeness}. Therefore Algorithm~\ref{alg:labeling} decides the bounded \ETOL model-checking problem correctly.
\end{proof}

\noindent Finally, the following theorem establishes the complexity of our model checking algorithm.
\begin{theorem}
\label{theoModelcheclTOL1}
%The model checking problem for \ETOL{} over \ac{TA} belongs to \PSPACE.
The bounded model-checking problem for \ETOL{} over \acp{TA}, with verification horizon $H$, belongs to \PSPACE.
\end{theorem}

\begin{proof}
Let $\TA$ be a timed automaton and let $\varphi$ be an \ETOL{} formula. We consider the bounded semantics up to the verification horizon $H$. The proof relies on an on-the-fly exploration of the bounded symbolic semantics. In particular, the complete bounded zone graph need not be stored explicitly. A symbolic state has the form
$\symZone=(\ell,Z)$, where $\ell\in L$ and $Z$ is a zone represented by a $\acs{dbm}$ over $C\cup\{\clockMC\}$. The observer clock $\clockMC$ is never reset and is constrained by $0\leq\clockMC\leq H$. Ordinary model clocks may be abstracted using the chosen finite zone
abstraction, whereas $\clockMC$ is kept exact so that equality of its values preserves equality of total execution durations. A \acs{dbm} contains a polynomial number of entries in the number of clocks. Moreover, each entry contains either infinity or an integer constant whose
binary representation is polynomial in the size of the input and of $H$. Consequently, a symbolic state can be represented using polynomial space. We first recall that a symbolic successor can be computed from the current
symbolic state by the standard zone operations: intersection with a guard, clock reset, time elapse, intersection with an invariant, canonicalisation,
and abstraction of the ordinary clocks. Each of these operations is computable in polynomial space. Hence successors of a symbolic state can be generated on demand without storing the whole symbolic graph. Since the bounded symbolic graph is finite and has at most exponentially many symbolic states, reachability between two symbolic states can be decided in polynomial space. Indeed, a nondeterministic algorithm only needs to store the current symbolic state and a counter bounded by the number of symbolic states. Thus bounded symbolic reachability is in $\NPSPACE$, and therefore in
$\PSPACE$ since $\NPSPACE=\PSPACE$. We now proceed by structural induction on $\varphi$. For atomic propositions and Boolean connectives, satisfaction can clearly be decided in polynomial space. Negation does not increase the space bound since $\PSPACE=\coPSPACE$. For the temporal operators, consider first
$\phi\until\psi$. By the induction hypothesis, membership in $X=\sat_H(\phi)$ and
$Y=\sat_H(\psi)$ can be checked in polynomial space. The procedure $\matchOU(X,Y)$ searches on the bounded symbolic graph for a finite symbolic execution satisfying the corresponding Until condition. Such an
execution can be guessed and checked on the fly while storing only the current symbolic state, the information required to verify membership in $X$ and $Y$, and a polynomial-size counter. Hence $\matchOU(X,Y)$
is computable in polynomial space. The same argument applies to $\phi\release\psi$. The procedure $matchR(X,Y)$ can be evaluated through an on-the-fly
search of the same bounded symbolic graph. Universal path conditions can be handled by complementation, using
$\PSPACE=\coPSPACE$. Therefore $\matchOR$ is also decidable in polynomial space. It remains to consider the execution-time opacity operators. Let $\until\subseteq\SymStates$. For $\Pre_{\ETexists}(U)$, by
\cref{prop:pre-etexists-correctness}, a symbolic state
$\symZone=(\ell,Z)$ belongs to $\Pre_{\ETexists}(U)$ iff there exists a concrete valuation represented by $Z$ and two non-empty executions starting from the same concrete state, one ending in $\gamma(U)$ and the other in its
complementary class, such that $\dur(\rho)=\dur(\rho')\leq H$. Because $\clockMC$ is never reset, equality of total durations is equivalent
to equality of the final observer-clock values:
$\dur(\rho)=\dur(\rho')
\Longleftrightarrow
\clockval_\rho(\clockMC)
=
\clockval_{\rho'}(\clockMC)$. Hence the two executions can be explored synchronously on the product of two bounded symbolic searches while storing only their current symbolic states and the corresponding \acs{dbm}s. This requires polynomial space. Thus
membership in $\Pre_{\ETexists}(U)$ is in $\NPSPACE$, and consequently in $\PSPACE$. For $\Pre_{\ETforall}(U)$, the condition requires that every execution ending
in $U$ admit a same-duration execution ending outside $U$, and conversely. A violation of this condition is witnessed by an execution on one side for
which no same-duration matching execution exists on the other side. The existence of such a counterexample can be checked using polynomial space by the same on-the-fly symbolic exploration. Since $\PSPACE$ is closed under
complement, $\Pre_{\ETforall}(\until)$ is also decidable in polynomial space. Consequently, each inductive case in the evaluation of an \ETOL{} formula can
be performed using polynomial space. The recursive evaluation need only store a polynomial number of symbolic states together with the current subformula;
intermediate satisfaction sets do not need to be materialised explicitly. Therefore, deciding whether the initial symbolic state satisfies $\varphi$ under the bounded semantics with horizon $H$ requires polynomial space in the size of the timed automaton, the formula, and the binary representation of $H$. Hence the bounded model-checking problem for \ETOL{} over timed automata
belongs to $\PSPACE$.
\end{proof}

\begin{table*}[th!]
\centering
\small
\begin{tabular}{l|c|c|c|c|c|c}
\hline
\cellHeader{Formula}
&
\cellHeader{Runs}
&
\cellHeader{Avg. (s)}
&
\cellHeader{Std. (s)}
&
\cellHeader{Min. (s)}
&
\cellHeader{Max. (s)}
&
\cellHeader{\ETOL{} satisfaction}
\\
\hline

$\prop_1$
& 50
& 0.047624
& 0.015133
& 0.042352
& 0.123596
& Satisfied (opaque)
\\
\hline

$\prop_2$
& 50
& 0.047549
& 0.016210
& 0.042159
& 0.129445
& Not satisfied (vulnerable)
\\
\hline

$\prop_3$
& 50
& 0.046580
& 0.015442
& 0.041415
& 0.121728
& Satisfied (opaque)
\\
\hline

$\prop_4$
& 50
& 0.047250
& 0.016072
& 0.041192
& 0.129162
& Not satisfied (vulnerable)
\\
\hline

$\prop_5$
& 50
& 0.047865
& 0.015648
& 0.042226
& 0.125722
& Not satisfied (vulnerable)
\\
\hline

$\prop_6$
& 50
& 0.048325
& 0.015715
& 0.042633
& 0.127642
& Not satisfied (vulnerable)
\\
\hline

$\prop_7$
& 50
& 0.046504
& 0.014156
& 0.041189
& 0.115293
& Satisfied (opaque)
\\
\hline

$\prop_8$
& 50
& 0.047317
& 0.014760
& 0.041619
& 0.121679
& Not satisfied (vulnerable)
\\
\hline
\end{tabular}
\caption{Repeated-run results for the ATM case study,  for verification horizon
$H=100$.}
\label{tab:atm-repeated-experiments}
\end{table*}

%%%%%%%%%%%%%%%%%%%%%%%%%%%%%%%%%%%%%%%%%%%%%%%%%%%%%%%%%%%%%%%%%%%%%%
%%%%%%%%%%%%%%%%%%%%%%%%%%%%%%%%%%%%%%%%%%%%%%%%%%%%%%%%%%%%%%%%%%%%%%
\section{Case study and Implementation}\label{sec:cs-implementation}
%%%%%%%%%%%%%%%%%%%%%%%%%%%%%%%%%%%%%%%%%%%%%%%%%%%%%%%%%%%%%%%%%%%%%%
%%%%%%%%%%%%%%%%%%%%%%%%%%%%%%%%%%%%%%%%%%%%%%%%%%%%%%%%%%%%%%%%%%%%%%
%\vspace{-0.2cm}
%%%%%%%%%%%%%%%%%%%%%%%%%%%%%%%%%%%%%%%%%%%%%%%%%%%%%%%%%%%%%%%%%%%%%%
\subsection{Case study}
\label{sec:casestudy}
%%%%%%%%%%%%%%%%%%%%%%%%%%%%%%%%%%%%%%%%%%%%%%%%%%%%%%%%%%%%%%%%%%%
 Consider the ATM system in \cref{fig:ATM}, used as a running case study for \ac{etopacity}. The model represents a typical user session, including password authentication, operation selection, cash withdrawal, balance inquiry, restart, cancellation, and termination. Transitions are labelled by actions such as $\ATMstart$, $\ATMaskPassword$, and $\ATMcorrectPassword$, and are equipped with timing constraints modelling user responses and system delays. We assume a timing attacker who observes only the total execution time, from the beginning of the session to a final state. The secret is whether cash is actually dispensed, represented by the action $\ATMtakeCash$. The model uses two clocks: a local clock $\sclock$ for interaction steps and a global clock $\sclocky$ used in the termination and
cancellation phases. Independently, for model checking we set the verification horizon to $H=100$ time units and that \(H\) is a parameter of the verification problem, not necessarily a constant imposed by the model itself. A session starts with $\ATMstart$, resets the clocks, and displays a welcome screen for $3$ time units before asking for a password. Incorrect passwords increment a finite failure counter; after three failures, or after a timeout of $10$ time units, the session is cancelled. After a correct password, the user may choose quick withdrawal, normal withdrawal, or balance inquiry. Withdrawal branches include a preparation phase of $15$ time units, and cash must be taken within $20$ time units. In the normal withdrawal branch, three invalid amount attempts or a timeout of $10$ time units also lead to cancellation. Finally, the user may restart or finish the session. This model is suitable for evaluating \ac{etopacity}, because different interaction patterns, such as quick withdrawal, normal withdrawal, repeated failures, balance inquiry, or early cancellation, may induce distinguishable execution times. We consider the following \ETOL{} properties. 

\begin{itemize} 
\item \textsf{Existential opacity of cash withdrawal.} There exists an execution time for which an observer cannot determine whether cash was dispensed: $ \prop_1 := \ETexists\bigl(\top \until \ATMtakeCash\bigr)$. 
\item \textsf{Universal bounded opacity of money availability.} Within the global session horizon, reaching a money-available state should remain indistinguishable from not reaching it: $ \prop_2 := j.\ETforall\bigl( \neg(\ATMMAQ \vee \ATMMAN) \until (\ATME \wedge j \leq 100) \bigr)$. 
\item \textsf{Local bounded opacity of cash taking.} There exists an execution time for which taking cash within a local delay remains indistinguishable: $\prop_3 := j.\ETexists\bigl( \top \until (\ATMtakeCash \wedge j \leq 20) \bigr)$.  
\item \textsf{Mixed-clock opacity.} The preparation phase and the final cash-taking action should remain ambiguous when both system time and specification time are considered: $\prop_4 := j.\ETforall\bigl( (\ATMPQW \wedge \sclock \leq 15) \until (\ATMtakeCash \wedge j \leq 25) \bigr)$.  
\item \textsf{Nested freeze clocks.} \ETOL{} can express opacity properties with several temporal reference points: %$\prop_5 := j_1.\bigl( (\neg \ATMC \wedge j_1 \leq 100) \land j_2.\ETexists\bigl( \ATMcorrectPassword \until (\ATMtakeCash \wedge j_2 \leq 30) \bigr) \bigr)$.  
%$\varphi_5 = j_1. \Bigl( \ETexists\bigl( \neg C\until (C \land j_1 \leq 100) \bigr) \land j_2. \ETexists\bigl( \psi_1  \until (\psi_2 \land j_2 \leq 10) \bigr) \Bigr).$
$\varphi_5 = j_1.(\ETexists\bigl(\neg \ATMC \until (\ATMC \land j_1 \leq 100) \bigr) \land j_2. \ETexists\bigl(\neg \ATMtakeCash \until (\ATME \land j_2 \leq 10) \bigr)).$
\item \textsf{Fixed-time termination opacity.} The formula checks ET-opacity of the property that
no cash-taking state is visited before reaching the terminal location $E$ exactly $10$ time units after the freeze point: $\prop_6 := j.\ETforall\bigl( \neg \ATMtakeCash \until (\ATME \wedge j = 10) \bigr)$. 
\item \textsf{Clocked opacity of the non-cash initial condition.} As positive sanity checks, we consider the existential and universal variants: $\varphi_7 =
j.\ETexists(\ATME \release (\neg \ATMtakeCash \land j \leq 100))$ and $\varphi_8 = j.\ETforall(\ATME \release
(\neg \ATMtakeCash \land j \leq 100)).$.

%$\prop_7 := j.\ETexists\bigl( (\neg \ATMtakeCash \wedge j \leq 100) \release (\neg \ATMtakeCash \wedge j \leq 100) \bigr)$ and  $\prop_8 := j.\ETforall\bigl( (\neg \ATMtakeCash \wedge j \leq 100) \release (\neg \ATMtakeCash \wedge j \leq 100) \bigr)$
%\item \textsf{Universal opacity of the non-cash invariant.} For every relevant observable execution time, the fact that cash has not yet been taken remains indistinguishable: $\prop_8 := \ETforall\bigl(\neg \ATMtakeCash \release \neg \ATMtakeCash\bigr)$.
\end{itemize}

\subsection{Implementation} %%%%%%%%%%%%%%%%%%%%%%%%%%%%%%%%%%%%%%%%%%%%%%%%%%%%%%%%%%%%%%%%%%%%%% 
We implemented the \ETOL{} model checking algorithm as a verification component integrated with the VITAMIN tool~\cite{AFV24}.\footnote{Source code available at: \url{https://github.com/jortizve/ETOL}.} The implementation is written in \texttt{Python}~3.11 and 
follows the bounded symbolic construction of
\cref{subsec:bounded-abstraction}  using \acp{DBM}~\cite{CLVO25}. It supports automaton clocks, freeze clocks, and the two ET-opacity operators $\ETexists$ and $\ETforall$ defined in \cref{def:sat}. The tool takes as input a model file and an \ETOL{} formula, builds the corresponding symbolic representation, and applies the model checking algorithm. \Acp{ta} are represented as graphs, implemented through adjacency matrices as specified in the input format. All experiments were run on a machine equipped with an Intel(R) Core(TM) i7-7700HQ CPU @ 2.80GHz, with 4 cores, 8 threads, and 16GB DDR4 RAM. 
%\vspace{-0.2cm}

\begin{figure}[t] 
\centering \includegraphics[width=0.90\linewidth]{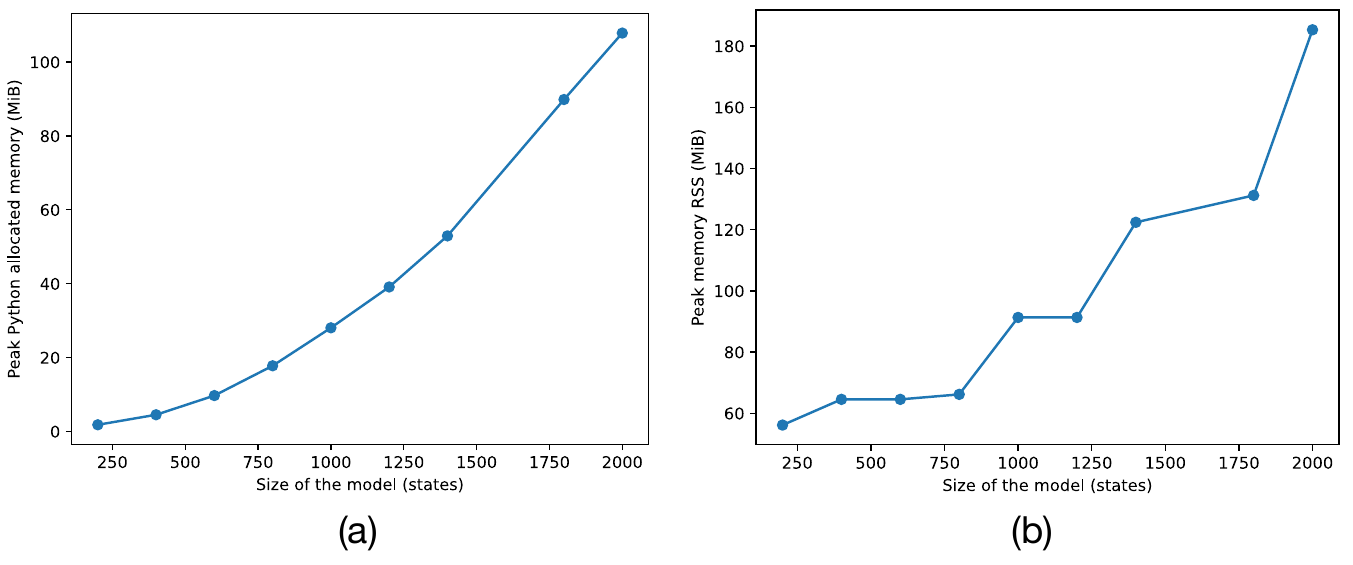} 
\vspace{-0.5cm} 
\caption{Peak memory consumption: (a) Python-allocated memory measured with
\texttt{tracemalloc}, and (b) process Resident Set Size (RSS).} \label{fig:mem-rss} \end{figure}

\subsection{Experiments} %%%%%%%%%%%%%%%%%%%%%%%%%%%%%%%%%%%%%%%%%%%%%%%%%%%%%%%%%%%%%%%%%%%%%% 

We evaluate the prototype along two dimensions. First, we use the ATM case study of \cref{fig:ATM} to validate the implementation on the formulas introduced in \cref{sec:casestudy}. Second, we perform scalability experiments on automatically generated \ac{TA} obtained by extending the ATM core. 

\noindent $\textsf{Proof \ of \ Concept: \ ATM \ Case \ Study.}$ We ran the checker on the ATM model of \cref{fig:ATM} and verified the eight properties $\prop_1$--$\prop_8$ defined in \cref{sec:casestudy}. For each property, the tool reports whether the initial symbolic state satisfies the formula and, when it is violated, may output witness executions illustrating the opacity violation. This experiment shows that the prototype returns both vulnerable and opaque verdicts on the ATM case study.
%\vspace{-0.5cm}

%\vspace{-0.8cm}

\noindent $\textsf{Repeated \ Runs \ on \ the \ ATM \ Case \ Study.}$ \label{subsec:atm-repeated-runs} To reduce measurement noise, each ATM experiment was repeated 50 times. For each formula, we report the average running time, standard deviation, minimum, and maximum execution time. The ATM model has 14 locations, 2 clocks, and 22 transitions. We evaluated the eight \ETOL{} formulas $\prop_1$--$\prop_8$ introduced in the case study, running the checker independently for each formula and recording the wall-clock time. The results are shown in \cref{tab:atm-repeated-experiments}. All executions terminated successfully. The average running time over the eight formulas is approximately
$0.0474$ seconds, and all formulas are verified in less than $0.05$ seconds on average. The small standard deviations indicate stable measurements. The column \ETOL{} satisfaction states whether the formula holds at the initial symbolic state. A result Not satisfied (vulnerable) means that the ATM violates the corresponding opacity property, while Satisfied (opaque) means that the property holds. Hence, vulnerability is interpreted property-wise.
 
 %The runtime distribution is shown in \cref{fig:atm-avg-runtime,fig:atm-runtime-boxplot}. The first plot reports the average runtime with standard deviation, while the second plot shows the distribution over the repeated executions.  These repeated measurements make the experimental protocol explicit and avoid relying on a single execution time. 

\noindent $\textsf{Scalability \ on \ Synthetic \ ATM \ Extensions.}$ %To assess scalability, we generated additional \ac{TA} parameterized by the number of locations, transitions, clocks, and maximal constants appearing in guards and invariants. These models are not copies of the handcrafted ATM model, instead, they are obtained by extending the ATM core with additional session modules, such as authentication, operation selection, withdrawal, and balance-inquiry branches. This construction increases the model size while preserving the relevant control states and atomic propositions used in the opacity specifications. Atomic propositions on newly added states are assigned automatically, while the original ATM labeling is kept unchanged. For each generated instance, we measure runtime and memory consumption. Memory is reported using two indicators: the peak Resident Set Size (RSS), which captures the overall process memory footprint, including Python objects, native allocations, and interpreter overhead; and the peak Python-allocated memory measured by \texttt{tracemalloc}, which accounts only for Python-level allocations. As shown in \cref{fig:mem-rss}, both measures increase smoothly with the number of states, indicating a controlled growth of the explored symbolic state space. The RSS curve is consistently higher than the \texttt{tracemalloc} curve, as expected, since RSS also includes the Python interpreter, external libraries, and non-Python memory. 
%To assess scalability, we generated additional \acp{TA} by extending the ATM core with extra session modules, such as authentication, operation selection, withdrawal, and balance-inquiry branches. The generated models vary in the number of locations, transitions, clocks, and maximal constants, while preserving the main ATM control states and atomic propositions used in the opacity specifications. New states are labelled automatically, whereas the original ATM labeling is kept unchanged. For each instance, we measure runtime and peak memory consumption. Memory is reported using both the peak Resident Set Size (RSS), which captures the total process footprint, and the peak Python allocation measured by \texttt{tracemalloc}. As shown in \cref{fig:mem-rss}, both measures grow smoothly with the number of states, suggesting a controlled increase of the explored symbolic state space. RSS is higher than \texttt{tracemalloc}, as it also includes interpreter overhead, external libraries, and native allocations. The extended benchmark suite, which includes additional case studies used to further validate our prototype, is reported in appendix.
To assess scalability, we generated additional \acp{TA} by extending the
ATM core with extra session modules, such as authentication, operation
selection, withdrawal, and balance-inquiry branches. The generated models
vary in size while preserving the main ATM control states and atomic
propositions used in the opacity specifications. New states are labelled
automatically, whereas the original ATM labelling is kept unchanged. For each generated instance, we measure runtime and peak memory consumption. Memory is reported using two complementary indicators: the peak Resident Set Size (RSS), which captures the overall memory footprint of the process, and the peak Python allocation measured by \texttt{tracemalloc}, which accounts only for memory allocated by Python-level objects. As shown in \cref{fig:mem-rss}, both measures exhibit an overall increasing trend as the number of states grows. The peak Python-allocated memory increases steadily from approximately $2$~MiB for the smallest instance to about $108$~MiB for the largest one. The RSS follows a less regular, stepwise growth, increasing from approximately $56$~MiB to $185$~MiB,
with several plateaus and larger increases for some model sizes. As expected, RSS remains higher than the memory reported by \texttt{tracemalloc}, since it additionally includes the Python interpreter, external libraries, native allocations, and other process overheads. These results indicate that memory consumption increases with the size of the explored symbolic model, while also showing that the two memory metrics capture different aspects of the implementation's resource usage. The extended benchmark suite, including additional case studies used to further validate the prototype, is reported in the appendix.

%\vspace{-0.3cm}
%%%%%%%%%%%%%%%%%%%%%%%%%%%%%%%%%%%%%%%%%%%%%%%%%%%%%%%%%%%%%%%%%%%%
%%%%%%%%%%%%%%%%%%%%%%%%%%%%%%%%%%%%%%%%%%%%%%%%%%%%%%%%%%%%%%%%%%%%
\section{Conclusion}\label{sec:conclusion}
%%%%%%%%%%%%%%%%%%%%%%%%%%%%%%%%%%%%%%%%%%%%%%%%%%%%%%%%%%%%%%%%%%%%
%%%%%%%%%%%%%%%%%%%%%%%%%%%%%%%%%%%%%%%%%%%%%%%%%%%%%%%%%%%%%%%%%%%%
In this work, we studied the verification of \ac{etopacity} in real-time systems modeled by \acp{ta}. Since timed opacity is undecidable for \acp{ta} in general~\cite{Cas09}, we focused on the decidable execution-time observation setting. We extended the location-based notion of \ac{etopacity} from~\cite{ADJ22} to a more general framework where secrets are specified by \ETOL formulas. We introduced \ETOL together with a symbolic bounded model-checking procedure based on zone exploration and exact execution-time matching. The observer clock recording elapsed time is kept exact, while finiteness is obtained by restricting the analysis to an explicit verification horizon. We establish soundness and completeness for this bounded setting and derive the corresponding \PSPACE{} upper bound. As future work, we plan to enrich \ETOL with quantitative, epistemic, strategic, and probabilistic features.

\ifdefined\VersionAuthor
    \newcommand{\AI}{Artificial Intelligence}
	\newcommand{\CCIS}{Communications in Computer and Information Science}
    \newcommand{\CACM}{Communications of the {ACM}}
    \newcommand{\CS}{{ACM} Computing Surveys}
    \newcommand{\DC}{Distributed Computing}
    \newcommand{\ENTCS}{Electronic Notes in Theoretical Computer Science}
    \newcommand{\ESE}{Empirical Software Engineering}
    \newcommand{\FAC}{Formal Aspects of Computing}
    \newcommand{\FI}{Fundamenta Informormaticae}
    \newcommand{\FMSD}{Formal Methods in System Design}
    \newcommand{\IJFCS}{International Journal of Foundations of Computer Science}
    \newcommand{\IJSSE}{International Journal of Secure Software Engineering}
    \newcommand{\IC}{Information and Computation}
    \newcommand{\IJC}{International Journal of Control}
    \newcommand{\IJIS}{International Journal of Information Security}
    \newcommand{\IPL}{Information Processing Letters}
    \newcommand{\IS}{Information Sciences}
    \newcommand{\JCS}{Journal of Computer Security}
    \newcommand{\JCSS}{Journal of Computer and System Sciences}
    \newcommand{\JISA}{Journal of Information Security and Applications}
    \newcommand{\JLAP}{Journal of Logic and Algebraic Programming}
    \newcommand{\JLAMP}{Journal of Logical and Algebraic Methods in Programming}
    \newcommand{\JLC}{Journal of Logic and Computation}
    \newcommand{\JSA}{Journal of Systems Architecture}
    \newcommand{\JSS}{Journal of Systems and Software}
    \newcommand{\LMCS}{Logical Methods in Computer Science}
    \newcommand{\LNCS}{Lecture Notes in Computer Science}
    \newcommand{\MSCS}{Mathematical Structures in Computer Science}
    \newcommand{\RESS}{Reliability Engineering \& System Safety}
    \newcommand{\SCP}{Science of Computer Programming}
    \newcommand{\STTT}{International Journal on Software Tools for Technology Transfer}
    \newcommand{\TAC}{{IEEE} Transactions on Automatic Control}
    \newcommand{\TCS}{Theoretical Computer Science}
    \newcommand{\TCL}{{ACM} Transactions on Computational Logic}
    \newcommand{\ToPNoC}{Transactions on Petri Nets and Other Models of Concurrency}
    \newcommand{\ToSEM}{{ACM} Transactions on Software Engineering and Methodology}
    \newcommand{\TCADICS}{{IEEE} Transactions on Computer-Aided Design of Integrated Circuits and Systems}
    \newcommand{\TSE}{IEEE Transactions on Software Engineering}
\else
    \newcommand{\AI}{AI}
    \newcommand{\CCIS}{CCIS}
    \newcommand{\CACM}{CACM}
    \newcommand{\CS}{CS}
    \newcommand{\DC}{DC}
    \newcommand{\ENTCS}{ENTCS}
    \newcommand{\ESE}{ESE}
    \newcommand{\FAC}{FAC}
    \newcommand{\FI}{FI}
    \newcommand{\FMSD}{FMSD}
    \newcommand{\IJFCS}{IJFCS}
    \newcommand{\IJSSE}{IJSSE}
    \newcommand{\IC}{IC}
    \newcommand{\IJC}{IJC}
    \newcommand{\IJIS}{IJIS}
    \newcommand{\IPL}{IPL}
    \newcommand{\IS}{IS}
    \newcommand{\JCS}{JCS}
    \newcommand{\JCSS}{JCSS}
    \newcommand{\JISA}{JISA}
    \newcommand{\JLAP}{JLAP}
    \newcommand{\JLAMP}{JLAMP}
    \newcommand{\JLC}{JLC}
    \newcommand{\JSA}{JSA}
    \newcommand{\JSS}{JSS}
    \newcommand{\LMCS}{LMCS}
    \newcommand{\LNCS}{LNCS}
    \newcommand{\MSCS}{MSCS}
    \newcommand{\RESS}{RESS}
    \newcommand{\SCP}{SCP}
    \newcommand{\STTT}{STTT}
    \newcommand{\TAC}{TAC}
    \newcommand{\TCS}{TCS}
    \newcommand{\TCL}{TCL}
    \newcommand{\ToPNoC}{ToPNoC}
    \newcommand{\ToSEM}{ACM ToSEM}
    \newcommand{\TCADICS}{TCADICS}
    \newcommand{\TSE}{TSE}
\fi
%
%\ifdefined\VersionAuthor
%	\renewcommand*{\bibfont}{\small}
%	\printbibliography[title={References}]
%\else
%    \ifdefined\IEEE
%    \bibliographystyle{IEEEtran}
%    \fi
%    \ifdefined\LLNCS
%    \bibliographystyle{splncs04} % abbrv
%    \fi
%    \ifdefined\KR
%    \bibliographystyle{kr} % abbrv
%    \fi
%    \ifdefined\AMMAS
%    \bibliographystyle{ACM-Reference-Format} % abbrv
    %\bibliography{ETOL-restricted}
%\fi
 \bibliographystyle{ACM-Reference-Format} % abbrv
\bibliography{ETOL-restricted}
\newpage
\appendix
\section*{Appendix}

\section{Model Checker for ETOL}
\label{sec:MimeticTools}
%%%%%%%%%%%%%%%%%%%%%%%%%%%%%%%%%%%%%%%%%%%%%%%%%%%%%%%%%%%%%%%%%%%%%%
%%%%%%%%%%%%%%%%%%%%%%%%%%%%%%%%%%%%%%%%%%%%%%%%%%%%%%%%%%%%%%%%%%%%%%

%\vspace{-5pt}
We developed a symbolic model checker for the logic \ETOL\footnote{Source code available at: \url{https://github.com/jortizve/ETOL}}, designed to verify timed opeacity properties over timed systems modeled as \ac{ta}. The verification algorithm is based on symbolic zone-based exploration, extended with support for reasoning about freeze clocks. 
The verification engine operates over a symbolic representation of time using \acp{DBM}, enabling efficient manipulation and comparison of clock valuations. 
The tool is implemented in \abbrv{Python} 3.11 and covers the full core of \ETOL, including freeze clocks and both variants of the execution-time opacity operator: 
\emph{existential} opacity ($\ETexists \ $) and \emph{universal} opacity ($\ETforall \ $). 
Intuitively, $\ETexists \  (\varphi)$ requires the existence of a time-matching alternative behaviour for at least one witness behaviour satisfying $\varphi$, whereas $\ETforall \ (\varphi)$ requires such time-matching alternatives for all relevant behaviours (according to the formal semantics of  \ETOL. 

\begin{table*}[th!]
\centering
\small
\begin{tabular}{p{0.51\textwidth}|c|c|c|c}
\hline
\cellHeader{Formula}
&
\cellHeader{Runs}
&
\cellHeader{Avg. (s)}
&
\cellHeader{Std. (s)}
&
\cellHeader{\ETOL{} verdict}
\\
\hline

$\prop_1 =
\ETexists\bigl(\top \until \ATMtakeCash\bigr)$
&
50
&
0.047624
&
0.015133
&
Satisfied (opaque)
\\
\hline

$\prop_2 =
j.\ETforall\bigl(
\neg(\ATMMAQ \vee \ATMMAN)
\until
(\ATME \wedge j\leq100)
\bigr)$
&
50
&
0.047549
&
0.016210
&
Not satisfied (vulnerable)
\\
\hline

$\prop_3 =
j.\ETexists\bigl(
\top
\until
(\ATMtakeCash \wedge j\leq20)
\bigr)$
&
50
&
0.046580
&
0.015442
&
Satisfied (opaque)
\\
\hline

$\prop_4 =
j.\ETforall\bigl(
(\ATMPQW \wedge \sclock\leq15)
\until
(\ATMtakeCash \wedge j\leq25)
\bigr)$
&
50
&
0.047250
&
0.016072
&
Not satisfied (vulnerable)
\\
\hline

$\prop_5 =
j_1.\bigl(
\ETexists(
\neg\ATMC
\until
(\ATMC\wedge j_1\leq100))
\land
j_2.\ETexists(
\neg\ATMtakeCash
\until
(\ATME\wedge j_2\leq10))
\bigr)$
&
50
&
0.047865
&
0.015648
&
Not satisfied (vulnerable)
\\
\hline

$\prop_6 =
j.\ETforall\bigl(
\neg\ATMtakeCash
\until
(\ATME\wedge j=10)
\bigr)$
&
50
&
0.048325
&
0.015715
&
Not satisfied (vulnerable)
\\
\hline

$\prop_7 =
j.\ETexists\bigl(
\ATME
\release
(\neg\ATMtakeCash\wedge j\leq100)
\bigr)$
&
50
&
0.046504
&
0.014156
&
Satisfied (opaque)
\\
\hline

$\prop_8 =
j.\ETforall\bigl(
\ATME
\release
(\neg\ATMtakeCash\wedge j\leq100)
\bigr)$
&
50
&
0.047317
&
0.014760
&
Not satisfied (vulnerable)
\\
\hline

\end{tabular}

\caption{ATM case study: repeated-run results for the eight
\ETOL{} properties with verification horizon $H=100$.}
\label{tab:atm-etol-results}
\end{table*}

\subsection{Algorithmic workflow.}
Given a model and a formula, the checker proceeds as follows:
\begin{enumerate}
    \item Symbolically explore reachable zones using DBMs~\cite{CLVO25}, while maintaining the global non-resetting clock for execution duration.
    \item Evaluate subformulas bottom-up on the zone graph.
    \item For opacity formulas of the form $\ETexists \,\varphi$  or $\ETforall \,\varphi$, construct (when possible) pairs of timed traces with identical duration: a \emph{witness trace} satisfying $\varphi$ and an \emph{alternative trace} of the same duration satisfying $\neg\varphi$.
\end{enumerate}
This trace-pair construction is used to justify opacity (or to exhibit a counterexample when it fails).

\subsection{Reproducible Evaluation by Scripts}
\label{subsec:etol-eval-scripts}

The experimental evaluation is fully reproducible through a script-based workflow (no graphical interface is required). 
We used:
\begin{itemize}
    \item \texttt{GeneratorATM.py}, which generates scalable ATM instances by extending a fixed ATM core with additional session modules, and
    \item \texttt{benchmark\_total\_mem\_time.py}, which runs the model checker on a list of \ETOL formulas and records execution-time and memory statistics (peak RSS and Python allocations), as well as the number of explored zones when available.
\end{itemize}

\noindent We executed the checker on multiple case studies, including the ATM scenario (detailed in  \abbrv{example 1}) and synthetic ATM-like families generated at increasing sizes. 
The remainder of this section reports representative results on the base ATM instance with larger-scale benchmarks and resource-usage plots. %During these runs, the tool printed warnings of the form \texttt{Element X not found in array.}  These messages indicate that some identifiers referenced by a formula were not present in the internal mapping produced by the current parser configuration (e.g., due to name normalization or model-format mismatches). These warnings do not affect the measurement methodology reported in \cref{tab:atm-etol-results}; nevertheless, they highlight the importance of consistent naming between model files and formulas.  We address these parsing/mapping issues in the supplementary material, together with corrected model encodings and extended benchmarks (including memory usage and explored-zone statistics).

\subsection{ATM Case Study: Results}
\label{subsec:atm-results}

\cref{tab:atm-etol-results} reports the outcomes of verifying a set of representative \ETOL formulas on the ATM model. 
For each formula, we provide the measured runtime and whether the formula holds in the initial state.
The tool additionally outputs the set of satisfying states computed for the formula.

\subsection{Example: Input file  (ATM Model)}
%%%%%%%%%%%%%%%%%%%%%%%%%%%%%%%%%%%%%%%%%%%%%%%%%%%%%%%%%%%%%%%%%%%

% \textbf{Input Model}

% \dm{C'est quel modèle ici? On ne laisse pas que ATM qui doit être suffisant?}

%  \begin{lstlisting}[language=java]
%  Transition
% 0 3 2 0 0 0
% 0 0 2 1 0 0
% 0 3 0 1 4 0
% 0 0 0 0 3 5
% 0 0 0 4 0 6
% 0 0 0 0 0 *
% Name_State
% s0 s1 s2 s3 s4 s5
% Initial_State
% s0
% Atomic_propositions
% r a
% Labelling
% 0 0
% 1 0
% 0 0
% 1 0
% 0 0
% 1 1
% Number_of_agents
% 1
% Clocks
% x y
% Clock_constraints
% 0 0 0 y>=1 0 0
% 0 0 0 0 0 0
% 0 0 y=0 0 0 0
% 0 0 x>4 0 y>1 x>2
% 0 0 0 0 0 x>2
% 0 0 0 0 0 0
% Invariants
% 0 0
% 0 0
% 0 0
% 0 0
% 0 x<=6
% 0 0
% \end{lstlisting}  

% \textbf{ETOL Properties}

% \begin{lstlisting}[language=java]
% y.O(!r U y <= 5)
% (*@$\OpOperator$@*)(!a | s0)
% x.O(!a U x >= 6)
% \end{lstlisting}  

% \textbf{ATM Input Model}

\begin{lstlisting}[language=etol, caption={ATM Input Model}, label=lst:atm-input]
    Transition
    0 1 0 0 0 0 0 0 0 0 0 0 0 0 0 0
    0 0 1 0 0 0 0 0 0 0 0 0 0 0 0 0
    2 2 3 0 0 0 0 0 0 0 0 0 0 0 0 0
    3 0 0 0 4 5 6 0 0 0 0 0 0 0 0 0
    0 0 0 0 0 0 0 7 0 0 0 0 0 0 0 0
    0 0 0 0 0 0 0 0 8 0 0 0 0 0 0 0
    0 0 0 0 0 0 0 0 0 9 0 0 0 0 0 0
    0 0 0 0 0 0 0 0 0 0 10 11 12 0 0 0
    0 0 0 0 0 0 0 0 0 0 0 0 0 13 0 0
    0 0 0 0 0 0 0 0 0 0 0 0 0 0 14 15
    0 0 0 0 0 0 0 0 0 0 0 0 0 0 0 16
    0 0 0 0 0 0 0 0 0 0 0 0 0 0 0 17
    0 0 0 0 0 0 0 0 0 0 0 0 0 0 0 18
    0 0 0 0 0 0 0 0 0 0 0 0 0 0 0 19
    0 0 0 0 0 0 0 0 0 0 0 0 0 0 0 20
    0 0 0 0 0 0 0 0 0 0 0 0 0 0 0 0

    Name_State
    I W WP WC WA PNW PQW DB MAN MAQ OO T C E

    Initial_State
    I

    Atomic_propositions
    cash balance

    Labelling
    0 0
    0 0
    0 0
    0 0
    0 0
    0 0
    0 0
    0 1
    1 0
    1 0
    0 0
    0 0
    0 0
    0 0

    Number_of_agents
    1

    Clocks
    x y

    Clock_constraints
    0 0 0 0 0 0 0 0 0 0 0 0 0 0 0 0
    x=3 0 0 0 0 0 0 0 0 0 0 0 0 0 0 0
    x<=10 0 x<=10 0 0 0 0 0 0 0 0 0 0 0 0 0
    0 0 0 x<=10 0 0 0 0 0 0 0 0 0 0 0 0
    0 0 0 0 x<=10 0 0 0 0 0 0 0 0 0 0 0
    0 0 0 0 0 x<=15 0 0 0 0 0 0 0 0 0 0
    0 0 0 0 0 0 x<=15 0 0 0 0 0 0 0 0 0
    0 0 0 0 0 0 0 x<=10 0 0 0 0 0 0 0 0
    0 0 0 0 0 0 0 0 x<=20 0 0 0 0 0 0 0
    0 0 0 0 0 0 0 0 0 x<=20 0 0 0 0 0 0
    0 0 0 0 0 0 0 0 0 0 x<=10 0 0 0 0 0
    0 0 0 0 0 0 0 0 0 0 0 y<=100 0 0 0 0
    0 0 0 0 0 0 0 0 0 0 0 0 y<=100 0 0
    0 0 0 0 0 0 0 0 0 0 0 0 0 0 0 0

    Invariants
    x<=3 0
    x<=10 0
    x<=10 0
    x<=10 0
    x<=15 0
    x<=15 0
    x<=10 0
    x<=20 0
    x<=20 0
    x<=10 0
    0 y<=100
    0 y<=100
    0 0

Verification_horizon
    H = 100
\end{lstlisting}

 \subsection{Input ATM Properties}

\begin{lstlisting}[language=etol, caption=ATM ETOL Properties, label=lst:atm-etol]

   OE(T U takeCash)

   j.OA(!(MAQ | MAN) U (E & j<=100))

   j.OE(T U (takeCash \& j<=20))

   j.OA((PQW \ x<=15) U (takeCash \& j<=25))

   j1.(OE(!C U (C \& j1<=100)) \&
       j2.OE(!takeCash U (E & j2<=10)))

   j.OA(!takeCash U (E & j=10))

   j.OE(E R (!takeCash & j<=100))

   j.OA(E R (!takeCash & j<=100))

\end{lstlisting}

\noindent The \cref{lst:atm-input} shows a snippet of the ATM model in a text format compatible with \ETOL. 
It defines the states, transitions, clocks, and invariants of the system. 
The \cref{lst:atm-etol} presents example \ETOL properties that can be verified against this model, focusing on \ac{etopacity} and timing constraints.  For each property, the tool reports whether the initial state satisfies the formula, and when applicable it outputs  a witness execution together with an alternative execution of identical total duration. The tool completed the verification within 0.047377 seconds.
%Detailed traces are summarized in \cref{tab:atm-traces}. %\paragraph{Synthetic timed models.}
To assess scalability, we conducted additional experiments on automatically generated \ac{ta}. %Each synthetic model is parameterized by : the number of locations, the number of transitions, the number of clocks, and (iv) maximal constants appearing in guards/invariants. For each size point,
We generate multiple random instances and measure the verification time of each  property. We report the average runtime over all generated instances, as well as standard deviation. Figure~\ref{fig:etol-synth-zones} plots execution time versus model size. %and (b) the number of explored symbolic zones versus model size, which is a primary driver of practical performance. Ò
\begin{figure}[t]
  \centering
  \includegraphics[width=0.9\linewidth]{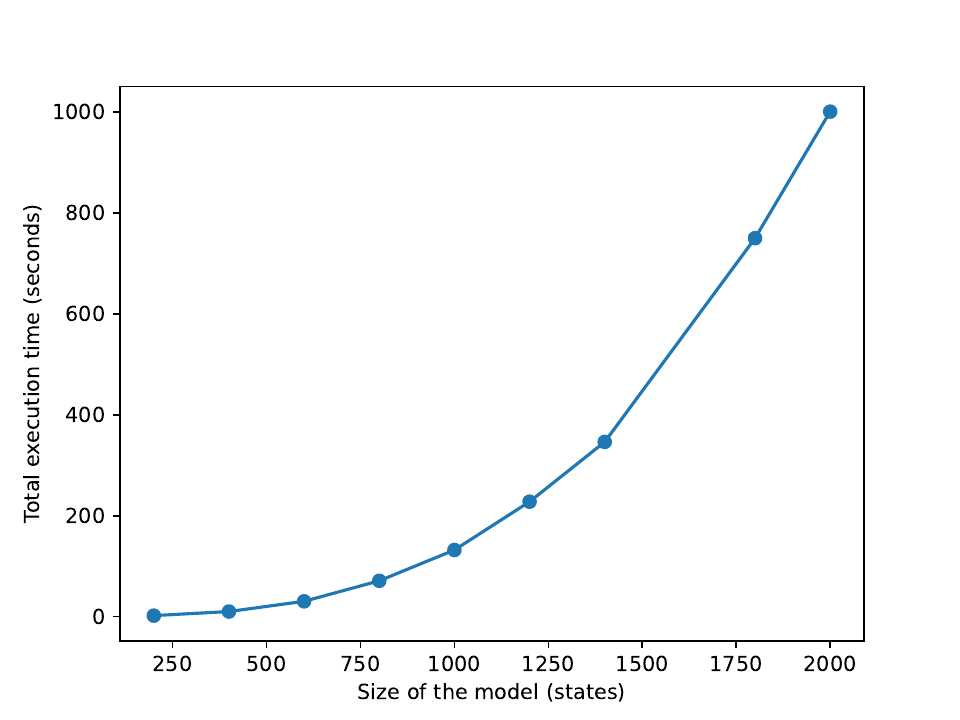}
  \caption{Tool execution times relative to the model size.}
  \label{fig:etol-synth-zones}
\end{figure}
In our synthetic ATM experiments, the analyzed models are not the handcrafted ATM instance of the case study: they are generated automatically by extending the ATM core with additional session modules (e.g., authentication, operation selection, and withdrawal/balance branches). This construction scales the system while preserving the semantics of the key control states and atomic propositions used in the opacity specifications. Model size is measured by the total number of states and transitions; atomic propositions on the added states are assigned automatically, while the original ATM labeling is kept unchanged, yielding diverse yet semantically consistent instances. As shown in Figure~\ref{fig:etol-synth-zones}, execution times exhibit the expected polynomial trend with respect to model size. For each reported size, we generated and verified a large set of random ATM instances, and the figure reports average running time. The largest models include over 1,000 states and a large transition relation, yet our zone-based opacity checking remains efficient, completing verification within a few seconds on standard hardware. 

\begin{figure}[t]
  \centering
  \includegraphics[width=0.9\linewidth]{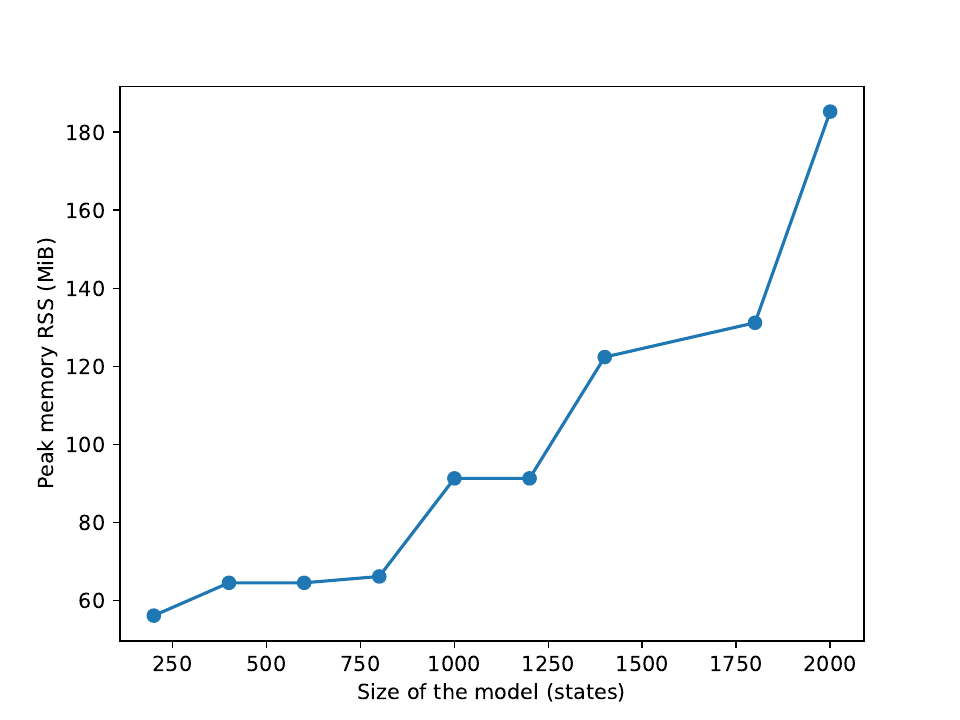}
  \caption{Tool execution memory relative to the model size.}
  \label{fig:mem-rss1}
\end{figure}

\noindent To complement runtime measurements, we report peak memory consumption during model checking using two indicators: (i) the peak resident set size (RSS), which captures the overall memory footprint of the process (including Python objects and native allocations), and (ii) the peak Python-allocated memory measured by \texttt{tracemalloc}, which reflects memory allocated by Python-level objects only. As shown in Fig.~\ref{fig:mem-rss1} and Fig.~\ref{fig:mem-py}, both measures increase smoothly with the number of states, confirming a controlled growth of the explored symbolic state space. The RSS curve is consistently higher than the \texttt{tracemalloc} curve, as expected, since RSS also includes interpreter overhead, libraries, and non-Python allocations. 

\begin{figure}[t]
  \centering
  \includegraphics[width=0.9\linewidth]{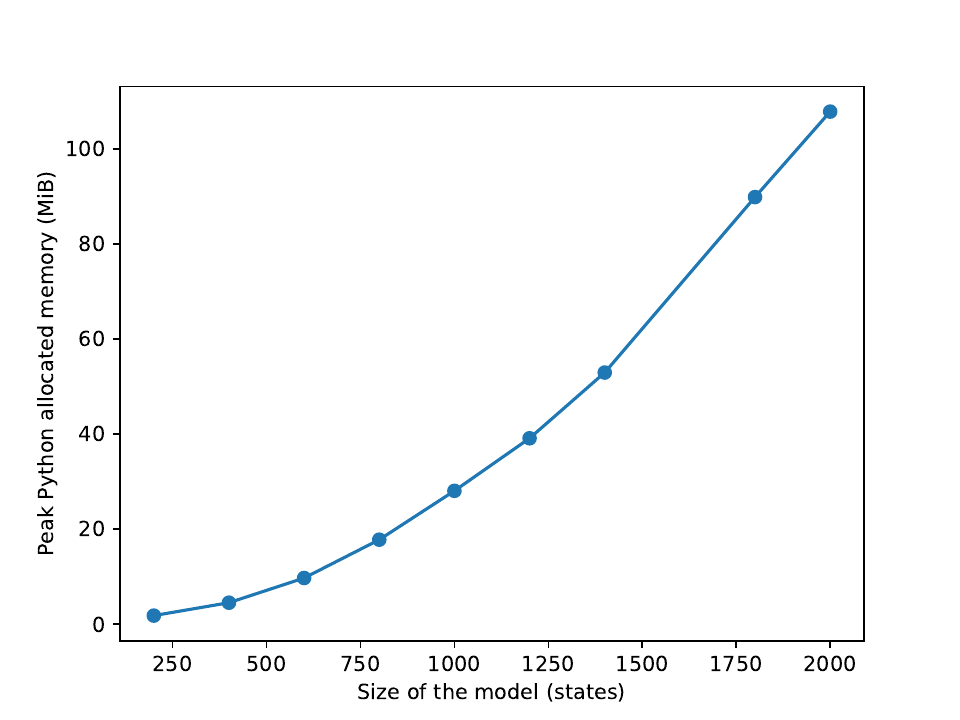}
  \caption{Tool execution memory relative to the model size.}
  \label{fig:mem-py}
\end{figure}

\subsection{Extended Benchmark Suite} To address the limitation of evaluating only variants of the ATM model, we extended the experimental campaign with additional case studies inspired by the timed-opacity benchmark suite of~\cite{ALMS22}. The suite includes examples from the timed-opacity and non-interference literature, such as the models of~\cite{GMR07,VNN18}, the Web privacy problem, the Coffee case study, Fischer's mutual exclusion protocol, and selected STAC timing-leak examples. Whenever possible, we translated the corresponding non-parametric models into the input format of our \ETOL{} prototype and checked the same execution-time opacity property. The results are reported in \cref{tab:extended-etol-benchmarks}. The columns $|S|$, $|X|$, and $|E|$ denote, respectively, the number of symbolic/control states, clocks, and transitions of the translated model. The column ``Runs'' gives the number of successful repeated executions. For terminating benchmarks, we report the average running time and standard deviation. If the prototype did not terminate within the selected timeout, the benchmark is marked as \emph{timeout}, and no correctness verdict is claimed. In the \ETOL{} verdict column, \emph{Valid/opaque} means that the checked \ETOL{} opacity formula is satisfied at the initial state, whereas \emph{False/vulnerable} means that the formula is not satisfied. 
%\begin{table*}[th!] 
%\centering \small 
%\begin{tabular}{l|c|c|c|c|c|c|c|c} 
%\hline \cellHeader{Benchmark} & \cellHeader{Source} & \cellHeader{$|S|$} & \cellHeader{$|X|$} & \cellHeader{$|E|$} & \cellHeader{Runs} & \cellHeader{Avg. (s)} & \cellHeader{Std. (s)} & \cellHeader{\ETOL{} verdict} \\ 
%\hline \cite{GMR07}-Fig.~1b & \cite{GMR07} & 3 & 2 & 3 & 38 & 0.00160 & 0.00010 & False/vulnerable \\ 
%\hline \cite{Ben+15a}-Web privacy & \cite{Ben+15a} & 9 & 3 & 11 & 14 & 0.00229 & 0.00015 & False/vulnerable \\ 
%\hline \cite{ALMS22}-Coffee & \cite{ALMS22} & 4 & 3 & 4 & 14 & 0.00213 & 0.00024 & Valid/opaque \\ 
%\hline \cite{HSRV02}-Fischer & \cite{HSRV02} & 14 & 1 & 15 & 14 & 0.00157 & 0.00009 & False/vulnerable \\
%\hline \cite{GMR07}-Fig.~2a & \cite{GMR07} & 4 & 2 & 4 & -- & timeout & -- & Unknown \\ 
%\hline \cite{GMR07}-Fig.~2b & \cite{GMR07} & 4 & 2 & 4 & -- & timeout & -- & Unknown \\ 
%\hline \cite{VNN18}-Fig.~5 & \cite{VNN18} & 2 & 2 & 1 & -- & timeout & -- & Unknown \\ 
%\hline \cite{ALMS22}-STAC & \cite{ALMS22} & 8 & 2 & 8 & -- & timeout & -- & Unknown \\ 
%\hline 
%\end{tabular} 
%\caption{Extended \ETOL{} benchmark results. The table reports the size of each translated model, the number of repeated successful runs, average execution time, standard deviation, and the verdict returned by the \ETOL{} checker. Timeout cases are reported explicitly and are not assigned a positive or negative verdict.} \label{tab:extended-etol-benchmarks} 
%\end{table*}

\begin{table*}[th!]
\centering
\small
\begin{tabular}{l|c|c|c|c|c|c|c|c}
\hline
\cellHeader{Benchmark}
&
\cellHeader{Source}
&
\cellHeader{$|S|$}
&
\cellHeader{$|X|$}
&
\cellHeader{$|E|$}
&
\cellHeader{Runs}
&
\cellHeader{Avg. (s)}
&
\cellHeader{Std. (s)}
&
\cellHeader{\ETOL{} verdict}
\\
\hline

\cite{GMR07}-Fig.~1b
& \cite{GMR07}
& 3 & 2 & 3
& 50
& 0.000109
& 0.000027
& Valid/opaque
\\
\hline

\cite{Ben+15a}-Web privacy
& \cite{Ben+15a}
& 9 & 3 & 11
& 50
& 0.000891
& 0.000127
& False/vulnerable
\\
\hline

\cite{ALMS22}-Coffee
& \cite{ALMS22}
& 4 & 3 & 4
& 50
& 0.000803
& 0.000187
& False/vulnerable
\\
\hline

\cite{HSRV02}-Fischer
& \cite{HSRV02}
& 14 & 1 & 15
& 50
& 0.000124
& 0.000019
& False/vulnerable
\\
\hline

\cite{GMR07}-Fig.~2a
& \cite{GMR07}
& 4 & 2 & 4
& 50
& 0.000141
& 0.000031
& Valid/opaque
\\
\hline

\cite{GMR07}-Fig.~2b
& \cite{GMR07}
& 4 & 2 & 4
& 50
& 0.000141
& 0.000048
& Valid/opaque
\\
\hline

\cite{VNN18}-Fig.~5
& \cite{VNN18}
& 2 & 2 & 1
& 50
& 0.000060
& 0.000015
& False/vulnerable
\\
\hline

\cite{ALMS22}-STAC
& \cite{ALMS22}
& 8 & 2 & 8
& --
& --
& --
& Not available
\\
\hline

\end{tabular}

\caption{Extended \ETOL{} benchmark results obtained with the bounded
on-the-fly implementation and verification horizon $H=20$. Each
available benchmark was executed independently 50 times.}
\label{tab:extended-etol-benchmarks}
\end{table*}

The extended benchmark suite evaluates the \ETOL{} implementation beyond
the ATM case study using examples from the timed-opacity and
non-interference literature, including Web privacy, Coffee, Fischer's
protocol, and related benchmark models. Since these benchmark formulas do
not contain an explicit global time bound, we use a common verification
horizon of $H=20$. This value is larger than the maximal timing constants
occurring in the bounded models considered here, while keeping the symbolic exploration compact. The ATM case study is evaluated separately with $H=100$, corresponding to its global session horizon. The results show that the implementation is not limited to the ATM model and can also analyse independent benchmarks with different control structures and timing constraints. They further indicate that verification cost is influenced not only by the number of states, but also by the structure of the symbolic state space and by the predecessor and fixed-point computations required by the checked property. Overall, these experiments strengthen the empirical evaluation by considering case studies from different sources and by reporting repeated executions with average running times and standard deviations.

\end{document}